\documentclass[12pt]{article}
\usepackage{amsmath,amsthm,amssymb}
\usepackage{mathrsfs}
\usepackage[T1]{fontenc}
\usepackage[latin9]{inputenc}
\usepackage{makeidx}
\usepackage{eufrak}
\begin{document}
\input amssym.def
\setcounter{equation}{0}
\newcommand{\wt}{{\rm wt}}
\newcommand{\spa}{\mbox{span}}
\newcommand{\Res}{\mbox{Res}}
\newcommand{\End}{\mbox{End}}
\newcommand{\Ind}{\mbox{Ind}}
\newcommand{\Hom}{\mbox{Hom}}
\newcommand{\Mod}{\mbox{Mod}}
\newcommand{\m}{\mbox{mod}\ }
\renewcommand{\theequation}{\thesection.\arabic{equation}}
\numberwithin{equation}{section}

\def \End{{\rm End}}
\def \Aut{{\rm Aut}}
\def \Z{\mathbb Z}
\def \H{\mathbb H}
\def \M{\Bbb M}
\def \C{\mathbb C}
\def \R{\mathbb R}
\def \Q{\mathbb Q}
\def \N{\mathbb N}
\def \ann{{\rm Ann}}
\def \<{\langle}
\def \o{\omega}
\def \O{\Omega}
\def \Or{\cal O}
\def \M{{\cal M}}
\def \1t{\frac{1}{T}}
\def \>{\rangle}
\def \t{\tau }
\def \a{\alpha }
\def \e{\epsilon }
\def \l{\lambda }
\def \L{\Lambda }
\def \g{\gamma}
\def \b{\beta }
\def \om{\omega }
\def \o{\omega }
\def \ot{\otimes}
\def \cg{\chi_g}
\def \ag{\alpha_g}
\def \ah{\alpha_h}
\def \ph{\psi_h}
\def \S{\cal S}
\def \nor{\vartriangleleft}
\def \V{V^{\natural}}
\def \voa{vertex operator algebra\ }
\def \voas{vertex operator algebras}
\def \v{vertex operator algebra\ }
\def \1{{\bf 1}}
\def \be{\begin{equation}\label}
\def \ee{\end{equation}}
\def \qed{\mbox{ $\square$}}
\def \pf {\noindent {\bf Proof:} \,}
\def \bl{\begin{lem}\label}
\def \el{\end{lem}}
\def \ba{\begin{array}}
\def \ea{\end{array}}
\def \bt{\begin{thm}\label}
\def \et{\end{thm}}
\def \br{\begin{rem}\label}
\def \er{\end{rem}}
\def \ed{\end{de}}
\def \bp{\begin{prop}\label}
\def \ep{\end{prop}}
\def \p{\phi}
\def \d{\delta}
\def \irr{\rm irr}

\newtheorem{th1}{Theorem}
\newtheorem{ree}[th1]{Remark}
\newtheorem{thm}{Theorem}[section]
\newtheorem{prop}[thm]{Proposition}
\newtheorem{coro}[thm]{Corollary}
\newtheorem{lem}[thm]{Lemma}
\newtheorem{rem}[thm]{Remark}
\newtheorem{de}[thm]{Definition}
\newtheorem{hy}[thm]{Hypothesis}
\newtheorem{conj}[thm]{Conjecture}
\newtheorem{ex}[thm]{Example}
\newtheorem{theorem}{Theorem}[section]
\newtheorem{corollary}[theorem]{Corollary}
\newtheorem{lemma}[theorem]{Lemma}
\newtheorem{proposition}[theorem]{Proposition}
\theoremstyle{definition}
\newtheorem{definition}[theorem]{Definition}
\theoremstyle{remark}
\newtheorem{remark}[theorem]{Remark}
\newtheorem{example}[theorem]{Example}
\numberwithin{equation}{section}

\begin{center}
{\Large {\bf Representations of Vertex Operator Algebras \\
\(V_{L_2}^{S_4}\), \(V_{L_2}^{A_5}\), and Their Quantum Dimensions}}\\
\vspace{0.5cm}

Li Wu\\
Chern Institute of Mathematics, \\
Key Lab of Pure Mathematics and Combinatorics
of Ministry of Education, \\
Nankai University, Tianjin 300071, People's Republic of China\\
nankai.wuli@gmail.com\\
Liuyi Zhang
\\
Department of Mathematics, \\
University of
California, Santa Cruz, \\
CA 95064 USA \\
lzhang19@ucsc.edu

\end{center}

\begin{abstract}
    \noindent \(C_2\) cofiniteness and rationality of \(V_{L_2}^{S_4}\) are obtained, and irreducible \(V_{L_2}^{S_4}\)-modules are classified. With the assumption of rationality and \(C_2\) cofiniteness, irreducible \(V_{L_2}^{A_5}\)-modules are determined.
    Also, quantum dimensions of these irreducible modules are calculated.
\end{abstract}

\section{Introduction}
\noindent Investigation of the vertex operator algebras \(V_{L_2}^{S_4}\) and \(V_{L_2}^{A_5}\) plays an important part in the classification of rational vertex operator algebras with \(c=1\). The vertex operator algebra \(V_{L_2}^{G}\) consists of \(G\) invariants of the even lattice vertex operator algebra \(V_{L_2}\). \\
\\
\noindent In the literature of physics at character level, references \cite{ginsparg1988curiosities} and \cite{kiritsis1989proof} studied classification of rational vertex operator algebras with \(c=1\).
These two references are based on two assumptions. one assumption is that the sum of the square norm of irreducible characters is invariant under the modular group. The second assumption is that each irreducible character is a modular function over a congruence subgroup.
Reference \cite{kiritsis1989proof} states that the character of a rational vertex operator algebra with \(c=1\) is the character of one of the vertex operators,
\begin{itemize}
    \item (a) lattice vertex operator algebras \(V_{L}\) associated with positive definite even lattices \(L\) of rank one,
    \item (b) orbifold vertex operator algebras \(V_{L}^{+}\) under the automorphism of \(V_L\) induced from the \(-1\) isometry of \(L\),
    \item (c)\(V_{L_2}^{G}\), where \(G\) is a finite subgroup of \(SO(3)\) isomorphic to one of \(\{A_4,S_4,A_5\}\).
\end{itemize}
\noindent Reference \cite{dong2010characterization} indicates that the list above would not be correct if the effective central charge \(\tilde{c}\), by Reference \cite{dong2004rational}, and the central charge \(c\) were different. Reference \cite{dong2004rational} characterizes the vertex operator algebra \(V_L\) for any positive definite even lattice \(L\) by \(c\), \(\tilde{c}\), and the rank of the weigh one subspace as a Lie algebra.
References \cite{dong2010characterization} \cite{dong2011characterization2011arXiv1110.1882D} \cite{dong2013characterization} \cite{zhang2009w} characterize the orbifold vertex operator algebra \(V_L^{+}\).
References \cite{Dong201376} \cite{dong2014characterization} \cite{Dong2015476} characterize the vertex operator algebra \(V_{L_2}^{A_4}\). Reference \cite{Dong201376} states the rationality and the \(C_2\)-cofiniteness of \(V_{L_2}^{A_4}\).
Since \(A_4\) is of order two in \(S_4\), the \(C_2\)cofiniteness of \(V_{L_2}^{A_4}\) implies the \(C_2\) cofiniteness of \(V_{L_2}^{S_4}\), by Reference \cite{miyamoto2009flatness}.\\
\\
Let \(V\) be a rational vertex operator algebra and \(G\) be a finite automorphism group of \(V\).
Orbifold theory conjecture indicates that \(V^G\) is rational and each irreducible \(V^G\)-module occurs in an irreducible \(g\)-twisted \(V\)-module for some \(g\in G\).
The second half of the conjecture is closed in Reference \cite{2015arXiv150703306DOrbifold}.
So, irreducible \(V_{L_2}^{S_4}\)-modules are classified. Also, the vertex operator algebra \(V_{L_2}^G\), where \(G\in \{A_4,S_4,A_5\}\),  is in the above list of rational vertex operator algebras. With the assumption of rationality and \(C_2\)-cofiniteness, irreducible modules of \(V_{L_2}^{A_5}\) can be determined. \\
\\
Reference \cite{dong2013quantum} defines and studies the quantum dimension and the globe dimension for a \(V\)-module. Also, Reference \cite{dong2013quantum} provides a quantum Galois theory \(V^G\subseteq V\), by References \cite{dong1997} \cite{hanaki1999quantum}. The results in Reference \cite{dong2013quantum} are strengthened in Reference \cite{2015arXiv150703306DOrbifold}, and are important in determining irreducible modules of \(V_{L_2}^{S_4}\) and \(V_{L_2}^{A_5}\).

\section{Preliminaries}

\noindent Let \((V,Y,\mathbf{1},\omega)\) be a vertex operator algebra and \(g\) an automorphism of \(V\) of fintie order \(T\).
Denote the decomposition of \(V\) into eigenspaces with respect to the action of \(g\) as
\begin{equation*}\label{g2.1}
V=\bigoplus_{r\in \Z/T\Z}V^r,
\end{equation*}
where $V^r=\{v\in V|gv=e^{-2\pi ir/T}v\}$.
Use $r$ to denote both
an integer between $0$ and $T-1$ and its residue class \m $T$ in this
situation.

\begin{de} \label{weak}
A {\em weak $g$-twisted $V$-module} $M$ is a vector space equipped
with a linear map
\begin{equation*}
\begin{split}
Y_M: V&\to (\End\,M)[[z^{1/T},z^{-1/T}]]\\
v&\mapsto\displaystyle{ Y_M(v,z)=\sum_{n\in\frac{1}{T}\Z}v_nz^{-n-1}\ \ \ (v_n\in
\End\,M)},
\end{split}
\end{equation*}
which satisfies the following:  for all $0\leq r\leq T-1,$ $u\in V^r$, $v\in V,$
$w\in M$,
\begin{eqnarray*}
& &Y_M(u,z)=\sum_{n\in \frac{r}{T}+\Z}u_nz^{-n-1} \label{1/2},\\
& &u_lw=0~~~
\mbox{for}~~~ l\gg 0,\label{vlw0}\\
& &Y_M({\mathbf 1},z)=Id_M,\label{vacuum}
\end{eqnarray*}
 \begin{equation*}\label{jacobi}
\begin{array}{c}
\displaystyle{z^{-1}_0\delta\left(\frac{z_1-z_2}{z_0}\right)
Y_M(u,z_1)Y_M(v,z_2)-z^{-1}_0\delta\left(\frac{z_2-z_1}{-z_0}\right)
Y_M(v,z_2)Y_M(u,z_1)}\\
\displaystyle{=z_2^{-1}\left(\frac{z_1-z_0}{z_2}\right)^{-r/T}
\delta\left(\frac{z_1-z_0}{z_2}\right)
Y_M(Y(u,z_0)v,z_2)},
\end{array}
\end{equation*}
where $\delta(z)=\sum_{n\in\Z}z^n$ and
all binomial expressions (here and below) are to be expanded in nonnegative
integral powers of the second variable.
\end{de}

\begin{de}\label{ordinary}
A $g$-{\em twisted $V$-module} is
a $\C$-graded weak $g$-twisted $V$-module $M:$
\begin{equation*}
M=\bigoplus_{\lambda \in{\C}}M_{\lambda}
\end{equation*}
where $M_{\l}=\{w\in M|L(0)w=\l w\}$ and $L(0)$ is the component operator of $Y(\omega,z)=\sum_{n\in \Z}L(n)z^{-n-2}.$ We also require that
$\dim M_{\l}$ is finite and for fixed $\l,$ $M_{\frac{n}{T}+\l}=0$
for all small enough integers $n.$\\
\\
\noindent If $w\in M_{\l}$, call $\l$ as the {\em weight} of
$w$ and write $\l=\wt w.$
\end{de}

\noindent Use $\Z_+$ to denote the set of nonnegative integers.
\begin{de}\label{admissible}
 An {\em admissible} $g$-twisted $V$-module
is a  $\frac1T{\Z}_{+}$-graded weak $g$-twisted $V$-module $M:$
\begin{equation*}
M=\bigoplus_{n\in\frac{1}{T}\Z_+}M(n)
\end{equation*}
satisfying
\begin{equation*}
v_mM(n)\subseteq M(n+\wt v-m-1)
\end{equation*}
for homogeneous $v\in V,$ $m,n\in \frac{1}{T}{\Z}.$
\ed
\noindent By Reference \cite{dong1997regularity}, if $g=Id_V$  these \(g\)-twisted notations become the notions of  weak, ordinary and admissible $V$-modules.\\
\\
\noindent If $M=\bigoplus_{n\in \frac{1}{T}\Z_+}M(n)$
is an admissible $g$-twisted $V$-module, the contragredient module $M'$
is defined as follows:
\begin{equation*}
M'=\bigoplus_{n\in \frac{1}{T}\Z_+}M(n)^{*},
\end{equation*}
where $M(n)^*=\Hom_{\C}(M(n),\C).$ The vertex operator
$Y_{M'}(a,z)$ is defined for $a\in V$ via
\begin{eqnarray*}
\langle Y_{M'}(a,z)f,u\rangle= \langle f,Y_M(e^{zL(1)}(-z^{-2})^{L(0)}a,z^{-1})u\rangle,
\end{eqnarray*}
where $\langle f,w\rangle=f(w)$ is the natural paring $M'\times M\to \C.$
It follows from References \cite{frenkel1993axiomatic}  and \cite{xu2000algebraic} that $(M',Y_{M'})$ is an admissible $g^{-1}$-twisted $V$-module.
Define the contragredient module $M'$ for a $g$-twisted $V$-module $M.$ In this case,
$M'$ is a $g^{-1}$-twisted $V$-module. Moreover, $M$ is irreducible if and only if $M'$ is irreducible.

\begin{de}
A \voa $V$ is called $g$-rational, if the  admissible $g$-twisted module category is semisimple. $V$ is called rational if $V$ is $1$-rational.
\end{de}

\noindent There is another important concept called $C_2$-cofiniteness, by Reference \cite{zhu1996modular}.
\begin{de}
A \voa $V$ is $C_2$-cofinite if $V/C_2(V)$ is finite dimensional, where $C_2(V)=\langle v_{-2}u|v,u\in V\rangle.$
\end{de}

\noindent The following results about $g$-rational \voas \  are well-known, by References \cite{dong1998twisted}, \cite{dong2000modular}.
\begin{thm}\label{grational}
If $V$ is $g$-rational, the following statements hold.
\begin{itemize}
    \item (1) Any irreducible admissible $g$-twisted $V$-module $M$ is a $g$-twisted $V$-module. Moreover, there exists a number $\l \in \mathbb{C}$ such that  $M=\oplus_{n\in \frac{1}{T}\mathbb{Z_+}}M_{\l +n}$ where $M_{\lambda}\neq 0.$ The $\l$ is called the conformal weight of $M;$
    \item (2) There are only finitely many irreducible admissible  $g$-twisted $V$-modules up to isomorphism.
    \item (3) If $V$ is also $C_2$-cofinite and $g^i$-rational for all $i\geq 0$ then the central charge $c$ and the conformal weight $\l$ of any irreducible $g$-twisted $V$-module $M$ are rational numbers.
\end{itemize}
\end{thm}
\noindent Let $V$ be a simple vertex operator algebra and $G$ a finite and
faithful group of automorphisms of $V$, and let $\mbox{Irr}\left(G\right)$
denote the set of simple characters $\chi$ of $G$. Now as $\mathbb{C}G$-module,
each homogeneous space $V_{n}$ of $V$ is of finite dimensional, and
so there is a direct sum decomposition of $V$ into graded subspaces
\[
V=\oplus_{\chi\in\mbox{Irr}G}V^{\chi},
\]
where $V^{\chi}$ is the subspace of $V$ on which $G$ acts according
to the character $\chi$. In other words, if $M_{\chi}$ is the simple
$\mathbb{C}G$-module affording $\chi$, then $V^{\chi}$ is the $M_{\chi}$-homogeneous
subspace of $V$ in the sense of group representation theory.

\begin{theorem} By \label{on quantum Galois theory} Reference \cite{dong2004rational}
 \inputencoding{latin1}, {suppose
that $V$ is a simple vertex operator algebra and that $G$ is a finite
and faithful solvable group of automorphisms of $V$. Then for $\chi\in\mbox{Irr}\left(G\right)$,
each $V^{\chi}$ is a simple module for the $G$-graded vertex operator
algebra $\mathbb{C}G\otimes V^{G}$ of the form
\[
V^{\chi}=M_{\chi}\otimes V_{\chi},
\]
where $M_{\chi}$ is the simple $\mathbb{C}G$-module affording $\chi$
and where $V_{\chi}$ is a simple $V^{G}$-module. }\inputencoding{latin9}\end{theorem}

\begin{theorem} By reference \cite{dong2014rationality}),
\inputencoding{latin1}{\label{DH}let $V$ be a simple
vertex operator superalgebra and $G$ a finite solvable subgroup of
$Aut\left(V\right)$. Suppose that $V^{G}$ is rational. Then $V$
is $g$-rational for any $g\in G$. }\inputencoding{latin9}\end{theorem}

\noindent For a $V$-module $M$ with grading $M=\oplus M_{n}$,
 define the formal character as
\[
\mbox{ch}_{q}M=q^{-\frac{c}{24}}\sum\dim M_{n}q^{n}=\mbox{tr}q^{-\frac{c}{24}+L\left(0\right)}.
\]
Denote the holomorphic function $\mbox{ch}_{q}M$ by $Z_{M}\left(\tau\right)$. Here and below, $\tau$ is in the upper half plane $\mathbb H$ and $q=e^{2\pi i\tau}$. \\
\\
\noindent Let $V$ be a rational, $C_{2}$-cofinite vertex operator algebra
with central charge $c$. Zhu, in Reference \cite{zhu1996modular}, proved that the space
\[
\left\langle q_{1}^{\left|a_{1}\right|}\cdots q_{n}^{\left|a_{n}\right|}\mbox{tr}_{W}Y\left(a_{1},q_{1}\right)\cdots Y\left(a_{n},q_{n}\right)q^{L\left(0\right)-c/24}:\ W\ \mbox{irreducible }V\mbox{-modules }\right\rangle
\]
is $SL_{2}\left(\text{\ensuremath{\mathbb{Z}}}\right)$-invariant
with $a_{i}\in V_{\left|a_{i}\right|}$, where $q_{j}=q_{z_{j}}=e^{2\pi iz_{j}}$
and $\left|a_{i}\right|$ denotes the weight of $a_{j}$. The concept
of $g$-twisted modules for a finite automorphism $g$ was introduced
in Reference \cite{dong2000modular} and the modular invariance of the space
\[
\left\langle \mbox{tr}_{M}g^{n}q^{L\left(0\right)-c/24}:n\in\mathbb{Z},\ M\ g\mbox{-twisted\ modules}\right\rangle
\]
was proved there. \\
\\
\noindent Let $M^{0},\cdots,M^{d}$ be the inequivalent irreducible $V$-modules
where $M^{0}\cong V$. Define
\[
Z_{i}\left(u,v,\tau\right)=\mbox{tr}_{M^{i}}e^{2\pi i\left(v\left(0\right)\right)+\left(v,u\right)\text{/2}}q^{L\left(0\right)+u\left(0\right)+\left(u,u\right)/2-c/24}
\]
for $u,v\in V_{1}$.

\begin{theorem} By Reference \cite{miyamoto1998representation}, let $V$ be a rational, $C_{2}$-cofinite vertex
operator algebra of CFT type. Assume $u,v\in V_{1}$ such that $u,v$
span an abelian Lie subalgebra of  $V_{1}$. Let $\gamma=\left(\begin{array}{cc}
a & b\\
c & d
\end{array}\right)\in SL\left(2,\mathbb{Z}\right).$ Then $Z_{i}\left(u,v,q\right)$ converges to a holomorphic function
in the upper half plane and
\[
Z_{i}\left(u,v,\gamma\tau\right)=\sum_{j=0}^{d}\gamma_{i,j}Z_{j}\left(au+bv,cu+dv,\tau\right),
\]
where $\gamma_{i,j}\in\mathbb{C}$ are independent of the choice
of $u,v,$ $\gamma\tau=\frac{a\tau+b}{c\tau+c}$.
\end{theorem}

\noindent The convergence of these functions has been established in Reference \cite{dong2005elliptic}.


\noindent The quantum dimensions of modules for vertex operator algebras are defined and their properties are discussed in Reference \cite{dong2013quantum}.
It is well known that $\mbox{qdim}_V M$ can be defined as the limit of $\frac{Ch_q M}{Ch_q V}$ as $q$ goes to 1 from the left.
The advantage of this definition is that one can use the modular transformation property of the $q$-characters, by Reference \cite{zhu1996modular}, and Verlinde formula, by References \cite{huang2008vertex} \cite{verlinde1988fusion}, to compute quantum dimensions and investigate their properties. The quantum dimensions for  rational and $C_2$-cofinite vertex operator algebras have some nice properties which enable us to determine fusion rules when the quantum dimensions can be calculated. The following are some properties of quantum dimensions, by Reference \cite{dong2013quantum}.

\begin{definition} \label{quantum dimension}Let $V$ be a vertex operator algebra and $M$
a $g$-twisted $V$-module such that $Z_{V}\left(\tau\right)$ and $Z_{M}\left(\tau\right)$
exists. The quantum dimension of $M$ over $V$ is defined as
\[
\mbox{qdim}_{V}M=\lim_{y\to0}\frac{Z_{M}\left(iy\right)}{Z_{V}\left(iy\right)},
\]
where $y$ is real and positive. \end{definition}

\noindent By Reference \cite{dong2000modular}, there is a natural action of Aut(\(V\)) on twisted modules. Let \(g,h\) be two automorphisms of \(V\) with \(g\) of finite order. If \(M, Y_g\) is a weak \(g\)-twisted \(V\)-module, there is a weak \(hgh^{-1}\) twisted \(V\) modules \((M\circ h), Y_{hgh^{-1}}\) where \(M\circ h \cong M\) as vector spaces and
\[
    Y_{hgh^{-1}}(v,z)=Y_g(h^{-1}v,z)
\]
for \(v\in V\). this defines a left action of Aut(\(V\)) on weak twisted \(V\)-modules and on isomorphism classes of weak twisted \(V\)-modules and on isomorphism classes of weak twisted \(V\)-modules. Symbolically, write
\[
    h\circ (M,Y_g)=(M\circ h, Y_{hgh^{-1}})=h\circ M.
\]
Sometimes abuse notation slightly by identifying \(M, Y_g\) with the isomorphism class that it defines. \\
\\
\noindent If \(g,h\) commute, obviously \(h\) acts on the \(g\)-twisted modules as above. Set \(\mathscr{M}(g)\) to be the equivalence classes of irreducible \(g\)-twisted \(V\)-modules and \(\mathscr{M}(g,h)=\{M\in \mathscr{M}(g)|h\circ M\cong M\}\). Then, for any \(M\in \mathscr{M}(g,h)\), there is a \(g\)-twisted \(V\)-module isomorphism
\[
    \varphi(h): h\circ M \rightarrow M.
\]
The linear map \(\varphi (h)\) is unique up to a nonzero scalar.

\begin{definition}
    By Reference \cite{dong2013quantum}, define the \emph{global dimension} of \(V\) as
    \[
        \mathrm{glob}(V)=\sum_{i=0}^{d}(\mathrm{qdim}_VM^{i})^2.
    \]
\end{definition}

\section{Basic results of irreducible modules}
In the rest of the paper, assume the following if not specified.
\begin{itemize}
    \item (V1) $V=\oplus_{n\geq 0}V_n$ is a simple vertex operator algebra of CFT type,
    \item (V2) $G$ is a finite automorphism group of $V$ and $V^G$ is a a vertex operator algebra of CFT type,
    \item (V3) $V^G$ is $C_{2}$-cofinite and  rational,
    \item (V4) The conformal weight of any irreducible $V^G$-module $N$  is nonnegative and is zero if and only if $N=V^G.$
\end{itemize}
\noindent Let \(V\) be a vertex operator algebra, \((W,Y)\)an irreducible \(V\)-module, and \(g\) an automorphism of \(V\).
    Definite a linear map
    \[
        Y^\sigma:V\rightarrow (\mathrm{End}\ W)[[z, z^{-1}]]
    \]
    by
    \[
        Y^\sigma(u,z)w=Y(\sigma^{-1}(u),z)w,
    \]
    where \(u \in V\), and \(w\in W\).
Reference \cite{dong1998twisted} shows that \((W,Y^\sigma)\) is an irreducible \(V\)-module. Denote \((W,Y^\sigma)\) by \(W^\sigma \).
\begin{definition}
    \label{defStableModule}
    By Reference \cite{dong1998twisted}, a \(V\)-module \(W\) is \(g\) \emph{stable} if \(W\cong W^\sigma\).
\end{definition}

\begin{theorem}\label{minvariance}
By Reference \cite{dong1998twisted}, the cardinalities $|\mathscr{M}(g,h)|$ and $|\mathscr{M}(g^ah^c,g^bh^d)|$ are equal for any $(g,h)\in P(V)$ and
$\gamma\in \Gamma.$ In particular, the number of irreducible $g$-twisted $V$-module is exactly the number of irreducible $V$-modules which are  $g$-stable.
\end{theorem}

\noindent The next two lemmas are from Reference \cite {dong1997}, and provide a practical way to construct irreducible \(V^G\) modules.

\begin{lemma}
    \label{constructionLemma1}
    Let \(V\) be a vertex operator algebra with an automorphism \(g\) of order \(T\).
    Let \(M=\sum_{n \in \frac{1}{T}\mathbb{Z}_+}M(n)\) be an irreducible \(g\)-twisted admissible \(V\)-Module.
    Then \(M^i=\bigoplus_{n\in \frac{i}{T}+\mathbb{Z}}M(n)\) is an irreducible \(V^{\langle g \rangle}\)-module for \(i=0,\ldots,T-1\).
\end{lemma}

\begin{lemma}
    \label{constructionLemma2}
    Let \(V\) be a simple vertex operator algebra, \(g\) an automorphism of \(V\) of prime order \(p\), and \(M\) an irreducible \(V\)-module such that \(g\circ M\) is not isomorphic to \(M\) as \(V\)-modules.
    Then, \(M\) is an irreducible \(V^{\langle g \rangle}\)-module.
\end{lemma}

\noindent The next two theorems are from Reference \cite{2015arXiv150703306DOrbifold}, and are the key theorems in determining irreducible modules of a rational, \(C_2\) cofinite vertex operator algebra.
\begin{theorem}
\label{MT1}
Let $g,h\in G$, $M$ an irreducible $g$-twisted $V$-module, $N$ an irreducible $h$-twisted $V$-module. Also assume that $M,N$ are not in the same orbit of $\mathscr {S}$ under the action of $G.$ Then
\\
\noindent 1) Each $M_{\lambda}$ for $\lambda \in \Lambda_{G_M,\alpha_M}$ is an irreducible $V^G$-module.
\\
\noindent 2) For any $\lambda\in \Lambda_{G_M,\alpha_M}$ and $\mu \in \Lambda_{G_N,\alpha_N},$ the irreducible $V^G$-modules
$M_{\lambda}$ and $N_{\mu}$ are inequivalent.
\end{theorem}

\begin{theorem}\label{MT2} Any irreducible $V^G$-module is isomorphic to an irreducible $V^G$-submodule $M_{\lambda}$ for some
irreducible $g$-twisted $V$-module $M$ and some $\lambda\in \Lambda_{G_M,\alpha_M}$.
\end{theorem}

\begin{remark}
    \label{remarkType}
    Theorem \ref{MT2} shows that there are two types of irreducible \(V^{G}\) modules modules.
    \begin{itemize}
        \item An irreducible \(V^{G}\) module \(M\) is of \emph{type one} if \(M\) occurs in the decomposition of irreducible \(V\) modules, as \(V^{G}\) modules.
        \item An irreducible \(V^{G}\) module \(M\) is of \emph{type two} if \(M\) does not occur in the decomposition of irreducible \(V\) modules, as \(V^{G}\) modules. That is, \(M\) occurs in a \(g\) twisted \(V^{G}\) module for some \(g\in G\) and \(g\neq 1\).
    \end{itemize}
\end{remark}

\begin{theorem}
    \label{MT3}
    Let \(V\) be a rational vertex operator algebra, and \(G\) an automorphism of \(V\). Let \(H\) be a subgroup of \(G\) such that \(C_G(h)=H\) for each \(h\in H\). Assume that \(V^H\) is rational. (Note that \(V^G\) is not necessarily to be rational). Then, an irreducible \(V^H\) module of type two is an irreducible \(V^G\) module of type two.
\end{theorem}

\begin{proof}
    Let \(M\) be an irreducible $V^{H}$ module of type two.
     Then it occurs in some \(g\)-twisted $V$ module \(N\) with $g\in H$.
     Let $G_N$ be the subgroup of $G$ consisting of $h\in G$ such that $N\circ h\cong N$.
     Note that $G_N\subset C_G{h}=H$.
     So, $G_N$ is the subgroup of $H$ consisting of $h\in H$ such that $N\circ h\cong N$.
     Let $(G_N,\alpha_N)$ be a projective representation of $G_N$ on $N$.
     Theorem \ref{MT1} shows that
     \begin{equation}
      N=\bigoplus_{\lambda\in irr(G_N,\alpha_N)} W_\lambda \otimes N_\lambda.
     \end{equation}
    Then $M=N_\lambda$ for some $\lambda$.
    Since each $N_\lambda$ is an irreducible $V^G$ module, \(M\) is an irreducible $V^G$ module.
\end{proof}

\begin{theorem}
    \label{Qd1}
    By Reference \cite{2015arXiv150703306DOrbifold}, one has
    $$\mathrm{qdim}_{V^G}M=|G|\mathrm{qdim}_VM.$$
    for any irreducible $g$-twisted $V$-module $M$.
\end{theorem}

\begin{theorem}\label{Qd2} By Reference \cite{2015arXiv150703306DOrbifold}, one has the following relation,
    \[
        \mathrm{glob}(V^G)=|G|^2\mathrm{glob}(V).
    \]
\end{theorem}

\begin{theorem}
    \label{qd3}
    By Reference \cite{Dong2015476}, let $V$ be a rational,
    $C_{2}$-cofinite vertex operator algebra with central charge $c$
    and $M^{0},\cdots,M^{d}$ be the inequivalent irreducible $V$-modules
    with $M^{0}\cong V$. Then $\mbox{qdim}M^{i}=\mbox{qdim}\left(M^{i}\right)^{\sigma}$.
\end{theorem}

\begin{theorem}
    \label{theoremConjugate}
    Let \(V\) be a vertex operator algebra, and \(G\) an automorphism group of \(V\). Let \(H_1\) and \(H_2\) be two subgroups of \(G\) such that \(g^{-1}H_1g=H_2\) for some \(g\in G\), that is, \(H_1\) and \(H_2\) are conjugate under \(G\). Let \(M_1\) be an irreducible \(V^{H_1}\) module. Then, there exists an irreducible \(V^{H_2}\) module, \(M_2\) such that as \(V_G\) modules, \(M_1\cong M_2\).
\end{theorem}
\begin{proof}
    Let \(h_2 \in H_2\), \(v \in V^{H_1}\). Then, there exists \(h_1 \in H_1\) such that \(g^{-1}h_1g=h_2\), that is \(g^{-1}h_1=h_2g^{-1}\)
    \begin{align*}
        h_2(g^{-1}v)&=g^{-1}h_1v\\
                        &=g^{-1}v.
    \end{align*}
    This shows that \(g^{-1}v \in V^{H_2}\). Hence, \(g^{-1}V^{H_1} \subseteq V^{H_2}\). Likewise, \(gV_{H_2}\subseteq V^{H_1}\). Thus,
    \begin{align*}
        V^{H_2}&=g^{-1}(gV^{H_2})\\
                 &\subseteq g^{-1}V^{H_1}\\
                 &\subseteq V^{H_2}.
    \end{align*}
    So, \(g^{-1}V^{H_1}=V^{H_2}\). \\
    \\
    \noindent Let \(M_1\circ g\cong M_1\) as vector spaces. Consider \((M_1\circ g, Y_{M_1\circ g})\) as a \(V^{H_2}\) module. For \(v\in V^{H_2}\), define
    \[
        Y_{M_1\circ g}(v,z):=Y_{M_1}(gv,z).
    \]
    The fact \(g^{-1}V^{H_1}=V^{H_2}\) shows that \(Y_{M_1\circ g}\) is well defined, because \(g\) is an automorphism of \(V\). Write
    \[
        (M_1, Y_{M_1})\circ g=(M_1\circ g, Y_{M_1\circ g})=M_1\circ g.
    \]
    \noindent Let \(N\) be a submodule of \(M_1\circ g\), as \(V^{H_2}\)  modules.
    Then, \(N\circ g^{-1}\) is a submodule of \(M_1\), as \(V^{H_1}\)  modules. The fact that \(M_1\) is an irreducible \(V^{H_1}\) module shows that \(N\circ g^{-1}\) is trivial in \(M_1\).
    Thus, \(N=((N\circ g^{-1})\circ g)\) is trivial in \(M_1\circ h\). So, \(M_1\circ g\) is an irreducible \(V^{H_2}\) module. \\
    \\
    Let \(w\in V^G\). Then, \(gw=w\), and hence \(Y_{G}(w,z)=Y_{G}(gw,z)\). That is, \(M_1\circ g \cong M_1\), as \(V_G\) modules.
\end{proof}
\begin{remark}
    \label{remarkConjugate}
    Denote by \(\mathscr{M}(H)\) the set of all irreducible modules of \(V^{H}\).
    Let \(\{H_i| i \in I\}\) be a conjugacy class of an automorphism group \(G\).
    Theorem \ref{theoremConjugate} shows that \(\mathscr{M}(H_i)\), for \(i\in I\), are the same, as \(V^G\) modules.
\end{remark}

\noindent Recall the definition of Schur covering group. The following materials are based on References \cite{zbMATH05622792} \cite{MR1205350} \cite{schur2001representation}.
\begin{definition}
A group homomorphism from $D$ to $G$ is said to be a Schur cover of the finite group $G$ if:
\begin{itemize}
    \item 1. the kernel is contained both in the center and the derived subgroup of $D$, and
    \item 2. amongst all such homomorphism, this $D$ has maximal size.
\end{itemize}
A group $D$ is a Schur covering group for $G$ if there is a Schur cover from $D$ to $G$.
\end{definition}
\noindent The Schur covers of the symmetric and alternating groups were classified.
\begin{theorem}
    The symmetric group of degree $n\geq 4$ has two isomorphic classes of Schur covers, both of order $2\cdot n!$.
    Then alternating group of degree $n$ has one isomorphic class of Schur covers, which has order $n!$ except when $n$ is 6 or 7.
\end{theorem}

\begin{theorem}
    \label{theoremSchur}
    For $n=4$, the Schur cover of the alternating group $A_4$ is given by $SL(2,3)\rightarrow PSL(2,3)\cong A_4$.
    The Schur covers of the symmetric group $S_4$ are $GL(2,3)$ and binary octahedral group.
    For $n=5$, the Schur cover of the alternating group $A_5$ is given by $SL(2,5)\rightarrow PSL(2,5)\cong A_5$.
\end{theorem}
\noindent There is a construction of Schur covering group for $A_4$.
Let $\pi:SU(2)\rightarrow SO(3)$ be the natural homomorphism.
There is a injective homomorphism $i:A_4\rightarrow SO(3)$ such that $\pi^{-1}(i(A_4))$ is a Schur covering group of $A_4$.\\
\\
\noindent Since $A_4$ has only one 3-dimensional irreducible representation, for any injective homomorphism $i:A_4\rightarrow SO(3)$,
$\pi^{-1}(i(A_4))$ is a Schur covering group of $A_4$.

\begin{lemma}
Let $G$ be a finite group such that $A_4\subset G$.
 Let $\pi:SU(2)\rightarrow SO(3)$ be the canonical homomorphism.
 Let $\phi:G\rightarrow SO(3)$ be a injective homomorphism.
 Then the kernel of $\phi$ is contained both in the center and the derived subgroup of $\pi^{-1}(G)$.
\end{lemma}
\begin{proof}
Note that $\mathrm{ker}(\pi)=\{\mathrm{Id},-\mathrm{Id}\}\subset SU(2)$ which is the center of $SU(2)$.
So $\ker(\pi)$ is contained in the center of $\pi^{-1}(\phi(G))$.
Since $A_4\subset G$, $\pi^{-1}(\phi(A_4))\subset \pi^{-1}(\phi(G))$.
Note that $\pi^{-1}(\phi(A_4))$ is a Schur covering group of $A_4$.
So $\ker(\pi)$ is contained in the derived subgroup of   $\pi^{-1}(\phi(A_4))$ and is also contained in the derived subgroup of $\pi^{-1}(\phi(G))$.
\end{proof}

\begin{remark}
    \label{remarkSchurA5}
    Let $\phi$ be a injective homomorphism from $A_5$ to $SO(3)$.
    Note that the the degree of Schur covering group of $A_5$ is 120 which is the same as the degree of the group $\pi^{-1}(\phi(A_5))$.
    By the definition of Schur covering group, $\pi^{-1}(\phi(A_5))$ is the unique Schur covering group of $A_5$.
\end{remark}

\section{The vertex operator algebra \(V_{L_2}^{A_4}\)}

\noindent Let $L=\mathbb{Z}\alpha$ be a positive definite even lattice of rank
one. That is $\left(\alpha,\alpha\right)=2k$ for some positive integer
$k$. $\mathfrak{h}=L\otimes_{\mathbb{Z}}\mathbb{C}$ and $\hat{\mathfrak{h}}_{\mathbb{Z}}$
the corresponding Heisenberg algebra; the bilinear form on $L$ or
$\mathfrak{h}$ is denoted $\left\langle \cdot,\cdot\right\rangle $.
Denote by $M\left(1\right)$ the associated irreducible module for
$\hat{\mathfrak{h}}_{\mathbb{Z}}$ such that the canonical central
element of $\hat{\mathfrak{h}}_{\mathbb{Z}}$ acts as $1$. Let $\mathbb{C}\left[L\right]$
be the group algebra of $L$ with a basis $e^{\alpha}$ for $\alpha\in L$.
Let $\beta\in\mathfrak{h}$ such that $\left\langle \beta,\beta\right\rangle =1$.
It was proved in References \cite{borcherds1986vertex} \cite{frenkel1989vertex} that there is a linear map
\[
V_{L}\to\left(\mbox{End}V_{L}\right)\left[\left[z,z^{-1}\right]\right]
\]
\[
v\mapsto Y\left(v,z\right)=\sum_{n\in\mathbb{Z}}v_{n}z^{-n-1}\ \ \left(v_{n}\in\mbox{End}V_{L}\right)
\]
such that $V_{L}=\left(V_{L},Y,\mathbf{1},\omega\right)$ is a simple
vertex operator algebra where $\mathbf{1}=1\otimes e^{0}.$
The dual lattice $L^{\circ}$ of $L$ is
\[
L^{\circ}=\left\{ \lambda\in\mathfrak{h}|\left(\alpha,\lambda\right)\in\mathbb{Z}\right\} =\frac{1}{2k}L.
\]
Then $L^{\circ}=\cup_{i=-k+1}^{k}\left(L+\lambda_{i}\right)$ is the
coset decomposition with $\lambda_{i}=\frac{i}{2k}\alpha$. Set $V_{L+\lambda_{i}}=M\left(1\right)\otimes\mathbb{C}\left[L+\lambda_{i}\right]$.
Then $V_{L+\lambda_{i}}$ for $i=-k+1$, $\cdots$, $k$ are the irreducible
modules for $V_{L}$.
Let $\theta$ be an automorphism of $\hat{L}$ such that $\theta\left(\alpha\right)=-\alpha$.
We define an automorphism of $V_{L}$, denote again by $\theta$,
such that
\[
\theta\left(u\otimes e^{\alpha}\right)=\theta\left(u\right)\otimes e^{-\alpha}\ \mbox{for}\ u\in M\left(1\right)\ \mbox{and}\ \alpha\in\hat{L}.
\]
Here the action of $\theta$ on $M\left(1\right)$ is given by
\[
\theta\left(\beta\left(n_{1}\right)\cdots\beta\left(n_{k}\right)\right)=\left(-1\right)^{k}\beta\left(n_{1}\right)\cdots\beta\left(n_{k}\right).
\]
The $\theta$-invariant $V_{L}^{+}$ of $V_{L}$ form a simple vertex
operator subalgebra and the $\left(-1\right)$-eigenspace $V_{L}^{-}$
is an irreducible $V_{L}^{+}$-module. Clearly $V_{L}=V_{L}^{+}\oplus V_{L}^{-}.$\\
\\
\noindent Let $\chi_{s}$ be a character of $L/2L$ such that $\chi_{s}\left(\alpha\right)=\left(-1\right)^{s}$
for $s=0,1$ and $T_{\chi_{s}}=\mathbb{C}$ the irreducible $L/2L$-module
with character $\chi_{s}$. It is well known that $V_{L}^{T_{s}}=M\left(1\right)\left(\theta\right)\otimes T_{\chi_{s}}$
is an irreducible $\theta$-twisted $V_{L}$-module. By References \cite{frenkel1989vertex} \cite{dong1994twisted},
denote the $\pm1$-eigenspaces of $V_{L}^{T_{s}}$ under $\theta$
by $\left(V_{L}^{T_{s}}\right)^{\pm}$.

\begin{theorem}
    \label{theoremV_L^+}
    By Reference \cite{dong1999representations}, any irreducible $V_{L}^{+}$-module is isomorphic
    to one of the following modules,
    \[
        V_{L}^{\pm},V_{\lambda_{i}+L}\left(i\not=k\right),V_{\lambda_{k}+L}^{\pm},\left(V_{L}^{T_{s}}\right)^{\pm}.
    \]
\end{theorem}

\begin{theorem} The quantum dimensions for all irreducible $V_{L}^{+}$-modules over $V_{L}^{+}$
are given by the following tables.
\begin{center}
\begin{tabular}{|c|c|c|c|c|c|}
\hline
 & $V_{L}^{+}$ & $V_{L}^{-}$ & $V_{L+\frac{r}{2k}\alpha}(1\leq r\leq k-1)$  & $V_{L+\frac{\alpha}{2}^{+}}$  & $V_{L+\frac{\alpha}{2}^{-}}$\tabularnewline
\hline
$\omega$ & 0 & 1 & \(\frac{r^2}{4k}\) & \(\frac{k}{4}\) & $\frac{k}{4}$\tabularnewline
\hline
$\mathrm{qdim}$ & 1 & 1 & 2 & 1 & 1\tabularnewline
\hline
\end{tabular}
\par\end{center}

\begin{center}
\begin{tabular}{|c|c|c|c|c|}
\hline
 & \(V_L^{T_1,+}\) & \(V_L^{T_1,-}\) & \(V_L^{T_2,+}\) & \(V_L^{T_2,-}\)\tabularnewline
\hline
$\omega$ & \(\frac{1}{16}\) & \(\frac{9}{16}\) & \(\frac{1}{16}\) & \(\frac{9}{16}\)\tabularnewline
\hline
$\mathrm{qdim}$ & \(k\) & \(k\) & \(k\) & \(k\)\tabularnewline
\hline
\end{tabular}
\par\end{center}
\end{theorem}

\begin{proof}
    Reference \cite{dong1999representations} shows the action of \(\omega\) on the first level of each irreducible module. The definition of \(V_L^{T_i,\pm}\), for \(i=1,2\), indicates
    \[
        \mathrm{qdim}_{V_{L}^{+}}V_L^{T_i,\pm}=\mathrm{qdim}_{V_{L}^{+}}V_{\mathbb{Z}\alpha}.
    \]
    Hence, the quantum dimensions listed are obtained by Theorem \ref{Qd1} and Theorem \ref{Qd2}.
\end{proof}

\noindent Let $L_{2}$ be the root lattice of type $A_{1}$ and $A_{4}$
the alternating group which is a subgroup of the automorphism group
of lattice vertex operator algebra $V_{L_{2}}$. Motivated by the
classification of rational vertex operator algebras with $c=1$, $V_{L_{2}}^{A_{4}}$
was studied in Reference \cite{Dong201376}. The $C_{2}$-cofiniteness and rationality
of $V_{L_{2}}^{A_{4}}$ are obtained, and the irreducible modules
are classified. They first realize $V_{\mathbb{Z}\alpha}^{G}$ as
$\left(V_{\mathbb{Z}\beta}^{+}\right)^{\left\langle \sigma\right\rangle }$
where $\left\langle \beta,\beta\right\rangle =8$ and $\sigma$ is
an automorphism of $sl\left(2,\mathbb{C}\right)$ of order 3, since
$V_{\mathbb{Z}\beta}^{+}$ is well understood, by  References \cite{abe2000fusion} \cite{abe2004rationality} \cite{abe2005rationality} \cite{dong1999representations} \cite{dong2001classification} \cite{Dong1999384}.\\
\\
\noindent Let $L_{2}=\mathbb{Z}\alpha$ be the rank one positive-definite even
lattice such that $\left(\alpha,\alpha\right)=2$ and $V_{L_{2}}$
the associated simple rational vertex operator algebra. Then $\left(V_{L_{2}}\right)_{1}\cong sl_{2}\left(\mathbb{C}\right)$
and $\left(V_{L_{2}}\right)_{1}$ has an orthonormal basis:
\[
x^{1}=\frac{1}{\sqrt{2}}\alpha\left(-1\right)\mathbf{1},\ x^{2}=\frac{1}{\sqrt{2}}\left(e^{\alpha}+e^{-\alpha}\right),\ x^{3}=\frac{i}{\sqrt{2}}\left(e^{\alpha}-e^{-\alpha}\right).
\]
\noindent Let $\tau_{i}\in Aut\left(V_{L_{2}}\right)$, $i=1,2,3$ be such that

\[
\tau_{1}\left(x^{1},x^{2},x^{3}\right)=\left(x^{1},x^{2},x^{3}\right)\left[\begin{array}{ccc}
1\\
 & -1\\
 &  & -1
\end{array}\right],
\]

\[
\tau_{2}\left(x^{1},x^{2},x^{3}\right)=\left(x^{1},x^{2},x^{3}\right)\left[\begin{array}{ccc}
-1\\
 & 1\\
 &  & -1
\end{array}\right],
\]

\[
\tau_{3}\left(x^{1},x^{2},x^{3}\right)=\left(x^{1},x^{2},x^{3}\right)\left[\begin{array}{ccc}
-1\\
 & -1\\
 &  & 1
\end{array}\right].
\]

\noindent Let $\sigma\in Aut\left(V_{L_{2}}\right)$ be such that
\[
\sigma\left(x^{1},x^{2},x^{3}\right)=\left(x^{1},x^{2},x^{3}\right)\left[\begin{array}{ccc}
0 & 1 & 0\\
0 & 0 & -1\\
-1 & 0 & 0
\end{array}\right].
\]\\
\\
\noindent Then $\sigma$ and $\tau_{i},i=1,2,3,$ generate a finite subgroup
of $\mathrm{Aut}\left(V_{L_{2}}\right)$ isomorphic to the alternating group
$A_{4}$. Denote this group by $A_{4}$. The subgroup $K$ generated by $\tau_{i},i=1,2,3$, is a normal
subgroup of $A_{4}$ of order $4$. Let $\beta=2\alpha$.

\begin{lemma}
     \label{lemmaK}
     By Reference \cite{dong1998rank}, $V_{L_{2}}^{K}\cong V_{\mathbb{Z}\beta}^{+}$.
\end{lemma}

\noindent Thus, $V_{L_{2}}^{A_{4}}=\left(V_{\mathbb{Z}\beta}^{+}\right)^{\left\langle \sigma\right\rangle }$.\\

\noindent Let $W^{1,T_{1}}$ and $W^{2,T_{1}}$ be the only two irreducible
$\sigma$-twisted modules of $V_{\mathbb{Z}\beta}^{+}$ and $W^{1,T_{2}}$,
$W^{2,T_{2}}$ be  the only two irreducible $\sigma^{2}$-twisted modules
of $V_{\mathbb{Z}\beta}^{+}$. Denote irreducible $V_{\mathbb{Z}\beta}^{+}$-submodules
of $W^{i,T_{j}}$ by $W^{i,T_{j},k}$ , $i,j=1,2$; $k=1,2,3$ which
are irreducible $\left(V_{\mathbb{Z}\beta}^{+}\right)^{\left\langle \sigma\right\rangle }$-modules.
Then there are exactly 21 irreducible modules of $\left(V_{\mathbb{Z}\beta}^{+}\right)^{\left\langle \sigma\right\rangle }$
which could be listed as following, by Reference \cite{Dong201376}.
\begin{align*}
    \{ \left(V_{\mathbb{Z}\beta}^{+}\right)^{m},V_{\mathbb{Z}\beta}^{-},
    &V_{\mathbb{Z}\beta+\frac{1}{8}\beta},V_{\mathbb{Z}\beta+\frac{3}{8}\beta},W^{i,T_{j},k},\left(V_{\mathbb{Z}\beta+\frac{1}{4}\beta}\right)^{n} \\
    &|m,n=0,1,2;i,j=1,2;k=1,2,3\}.
\end{align*}

\begin{lemma}
    \label{adtype}
    By Reference \cite{Dong2015476}, the group \(SO(3)\) is the connected compact subgroup of \(\mathrm{Aut}(V_{L_2})\), whose discrete subgroup are the cyclic group \(Z_n\), the dihedral group \(D_n\), \(A_4\), \(S_4\) and \(A_5\). Also, the vertex operator algebra \(V_{L_2}^{\mathbb{Z}_n}\cong V_{\mathbb{Z}n\alpha}\), and \(V_{L_2}^{D_n}\cong V_{\mathbb{Z}n\alpha}^+\).
\end{lemma}

\begin{lemma}
    \label{decompSpace}
    By Reference \cite{dong1997}, let \(g\) be an automorphism of the vertex operator algebra \(V_{L_2}\) and of order \(T\). Let \(\nu=T\alpha\). Then, eigenvalues of the \(g\) action on \(V_{\mathbb{Z}}\) show that
    \[
        V_{L_2}\cong \bigoplus_{i \in \mathbb{Z},\  0\leq i \leq T-1}V_{\mathbb{Z}\nu+\frac{i}{T}\nu},
    \]
and
    \[
        V_{L_2+\frac{1}{2}\alpha}\cong \bigoplus_{i\in \mathbb{Z},\  0\leq i \leq T-1}V_{\mathbb{Z}\nu-\frac{\nu}{2T}+\frac{i}{T}\nu}.
    \]
\end{lemma}

\begin{remark}
    \label{Mod0}
    Let \(g\) be an automorphism of the vertex operator algebra \(V_{L_2}\) and of order \(T\). Let \(\xi=T\alpha\).
    Notice that \(V_{L_2}\cong V_{L_2}^g\) and \(V_{L_2+\frac{\alpha}{2}}\cong V_{L_2+\frac{\alpha}{2}}^g\), that is, \(V_{L_2}\) and \(V_{L_2+\frac{\alpha}{2}}\) are \(g\) stable.
    Likewise, \(V_{L_2}\) and \(V_{L_2+\alpha/2}\) are \(g^i\) stable, where \(i \in \mathbb{Z},\ 0\leq i \leq T-1\).
    Theorem \ref{minvariance} shows that there are exactly two irreducible \(g^i\) twisted modules, where \(i \in \mathbb{Z},\ 0\leq i \leq T-1\).
    Lemma \ref{decompSpace} and Theorem \ref{MT1} show that \(V_{\mathbb{Z}\xi+\frac{r}{T}\xi}\) for \(r = 0\ (\mathrm{mod}\ T)\) are irreducible \(V_{L_2}^{\langle g \rangle}\) modules occurring in \(g^0\) twisted \(V_{L_2}\) modules.
    Hence, \(V_{\mathbb{Z}\xi+\frac{r}{T}\xi}\) for \(r \neq 0\ (\mathrm{mod}\ T)\) are irreducible \(V_{L_2}^{\langle g \rangle}\) modules occurring in \(g^r\) twisted \(V_{L_2}\) modules, where \(r\in \mathbb{Z},\ 1\leq r \leq T-1\).
\end{remark}

\begin{theorem}
    Irreducible modules of \(V_{L_2}^{A_4}\) are
    \[
        (V^+_{\mathbb{Z}\beta})^0, \ (V^+_{\mathbb{Z}\beta})^1, (V^+_{\mathbb{Z}\beta})^2,
    \]
    \[
        V^-_{\mathbb{Z}\beta}, \ V_{\mathbb{Z}\beta+\frac{1}{8} \beta }, \ V_{\mathbb{Z}\beta+\frac{3}{8} \beta },
    \]
    \[
        V_{\mathbb{Z}\beta+\frac{1}{4} \beta }^0, \ V_{\mathbb{Z}\beta+\frac{1}{4} \beta }^1, \ V_{\mathbb{Z}\beta+\frac{1}{4} \beta }^2,
    \]
    \[
        V_{\mathbb{Z}\gamma \pm \frac{r}{18}\gamma}, \ \mathrm{for} \  r\in \mathbb{Z}, \ 1\leq r \leq 8,\  \mathrm{and}\ r \neq 0 (\mathrm\ {mod}\ 3).
    \]
\end{theorem}

\noindent Reference \cite{Dong201376} gives twenty one irreducible modules of \(V_{L_2}^{A_4}\) by analyzing the action of \(\sigma\) on irreducible modules of \(V_{\mathbb{Z}\beta}^{+}\), and by constructing all irreducible \(\sigma^i\), for \(i=1,2\), twisted \(V_{\mathbb{Z}\beta}^{+}\) modules.
Reference \cite{Dong201376} shows that twelve of the twenty one irreducible modules come from irreducible \(\sigma^i\), for \(i=1,2\), twisted \(V_{\mathbb{Z}\beta}^+\) modules. That is, they are of type two.
These twelve irreducible modules are found from irreducible \(V_{\mathbb{Z}\gamma}\) modules, by Reference \cite{Dong2015476}, without specifying twisted modules. \\
\\
\noindent Let \(H\)=\(\langle \sigma \rangle\). Then, \(H\subseteq A_4\) satisfies the assumption in Theorem \ref{MT3}.
Remark \ref{Mod0} shows that \(V_{\mathbb{Z}\gamma \pm \frac{r}{18}}\), for \( r\in \mathbb{Z}, \ 1\leq r \leq 8\), and \(r \neq 0 \ (\mathrm{mod} \ 3)\) are twelve irreducible \(V_{L_2}^{\langle \sigma \rangle}\) modules of type two.
Theorem \ref{MT3} indicates that the twelve irreducible \(V_{L_2}^{\langle \sigma \rangle}\) modules of type two are twelve irreducible \(V_{L_2}^{A_4}\) modules of type two, probably isomorphic under \(V_{L_2}^{A_4}\).
Also, theorem \ref{MT2} indicates that the twelve irreducible \(V_{L_2}^{\langle \sigma \rangle}\) modules of type two exhaust irreducible \(V_{L_2}^{A_4}\) modules of type two.
Thus, Theorems \ref{Qd1} and \ref{Qd2} show that the twelve irreducible \(V_{L_2}^{\langle \sigma \rangle}\) modules of type two are exactly the twelve irreducible \(V_{L_2}^{A_4}\) modules of type two.
That is, they are  nonisomorphic under \(V_{L_2}^{A_4}\). The same idea will be used to find irreducible modules of \(V_{L_2}^{S_4}\).

\section{Irreducible modules of \(V_{L_2}^{S_4}\)}

\noindent Let \(\zeta=4\alpha\), and \(\rho\) an automorphism of \(V_{L_2}\) such that \(\rho(x^1)=-x^1\), \(\rho(x^2)=x^3\), and \(\rho(x^2)=x^3\). That is,
\[
    \rho \left(x^{1},x^{2},x^{3}\right)=\left(x^{1},x^{2},x^{3}\right)\left[\begin{array}{ccc}
        -1 & 0 & 0\\
        0 & 0 & 1\\
        0 & 1 & 0
    \end{array}\right],
\]
Then, \(\rho \in S_4 \backslash A_4\), and \(V_{L_2}^{S^4}\cong (V_{L_2}^{A_4})^{\langle \rho \rangle}\).
Definition of automorphism on vertex operator algebra shows that \(\rho (\omega)=\omega\). Observe the table in the next theorem, by Reference \cite{Dong2015476}
\begin{theorem}
    By Reference \cite{miyamoto2015c_2OfCyclicCMP}, let \(V\) be a \(C_2\)-cofinite simple vertex operator algebra of CFT-type and \(\sigma\in \mathrm{Aut}(V)\) of finite order \(p\). Then a fixed point vertex operator subalgebra \(V^{\sigma}\) is also \(C_2\)-cofinite.
\end{theorem}

\begin{remark}
    The preceding theorem shows that the vertex operator algebra \(V_{L_2}^{S^4}\cong (V_{L_2}^{A_4})^{\langle \rho \rangle}\) is \(C_2\) cofinite. Reference \cite{Dong201376} shows that \(V_{L_2}^{A_4}\) is rational. Also, the simpleness of \(V_{L_2}^{A_4}\)) shows that \(V_{L_2}^{A_4}\cong (V_{L_2}^{A_4})'\)). Thus, Condition I in Reference \cite{miyamoto2011Flat&Semi-Rigidity} is satisfied. Hence, Corollary 28 in Reference \cite{miyamoto2011Flat&Semi-Rigidity} indicates that \(V_{L_2}^{S^4}\cong (V_{L_2}^{A_4})^{\langle \rho \rangle}\) is rational.
\end{remark}

\begin{theorem} The quantum dimensions for all irreducible $\left(V_{\mathbb{Z}\beta}^{+}\right)^{\left\langle \sigma\right\rangle }$-modules
are given by the following tables $1-4$.
\begin{center}
\begin{tabular}{|c|c|c|c|c|c|c|}
\hline
 & $\left(V_{\mathbb{Z}\beta}^{+}\right)^{0}$ & $\left(V_{\mathbb{Z}\beta}^{+}\right)^{1}$ & $\left(V_{\mathbb{Z}\beta}^{+}\right)^{2}$  & $V_{\mathbb{Z}\beta}^{-}$  & $V_{\mathbb{Z}\beta+\frac{1}{8}\beta}$ & $V_{\mathbb{Z}\beta+\frac{3}{8}\beta}$\tabularnewline
\hline
$\text{\ensuremath{\omega}}$ & 0 & 4 & 4 & 1 & $\frac{1}{16}$ & $\frac{9}{16}$\tabularnewline
\hline
$\mbox{qdim}$ & $1$ & 1 & 1 & 3 & 6 & 6\tabularnewline
\hline
\end{tabular}
\par\end{center}

\begin{center}
\begin{tabular}{|c|c|c|c|c|c|c|}
\hline
 & $W^{1,T_{1},0}$ & $W^{1,T_{1},1}$  & $W^{1,T_{1},2}$ & $W^{2,T_{1},0}$ & $W^{2,T_{1},1}$ & $W^{2,T_{1},2}$\tabularnewline
\hline
$\omega$ & $\frac{1}{36}$ & $\frac{25}{36}$ & $\frac{49}{36}$ & $\frac{1}{9}$ & $\frac{4}{9}$ & $\frac{16}{9}$\tabularnewline
\hline
$\mbox{qdim}$ & 4 & 4 & 4 & 4 & 4 & 4\tabularnewline
\hline
\end{tabular}
\par\end{center}

\begin{center}
\begin{tabular}{|c|c|c|c|c|c|c|}
\hline
 & $W^{1,T_{2},0}$ & $W^{1,T_{2},1}$ & $W^{1,T_{2},2}$ & $W^{2,T_{2},0}$ & $W^{2,T_{2},1}$ & $W^{2,T_{2},2}$\tabularnewline
\hline
$\omega$ & $\frac{1}{36}$ & $\frac{25}{36}$ & $\frac{49}{36}$ & $\frac{1}{9}$ & $\frac{4}{9}$ & $\frac{16}{9}$\tabularnewline
\hline
$\mbox{qdim}$ & 4 & 4 & 4 & 4 & 4 & 4\tabularnewline
\hline
\end{tabular}
\par\end{center}

\begin{center}
\begin{tabular}{|c|c|c|c|}
\hline
 & $\left(V_{\mathbb{Z}\beta+\frac{1}{4}\beta}\right)^{0}$ & $\left(V_{\mathbb{Z}\beta+\frac{1}{4}\beta}\right)^{1}$ & $\left(V_{\mathbb{Z}\beta+\frac{1}{4}\beta}\right)^{2}$\tabularnewline
\hline
$\omega$ & $\frac{1}{4}$ & $\frac{9}{4}$ & $\frac{9}{4}$\tabularnewline
\hline
$\mbox{qdim}$ & 2 & 2 & 2\tabularnewline
\hline
\end{tabular}
\par\end{center}
\end{theorem}

\noindent The fact that \(\rho (\omega)=\omega\) shows that the irreducible modules with distinguished \(\omega\) actions are \(\rho\) stable. That is,  \((Z_{\mathbb{Z}\beta}^{+})^{0}\), \(V_{\mathbb{Z}\beta}^{-}\), \(V_{\mathbb{Z}\beta+\frac{1}{8}\beta}\), \(V_{\mathbb{Z}\beta+\frac{3}{8}\beta}\), \((V_{\mathbb{Z}\beta+\frac{1}{4}\beta})^0\), are \(\rho\) stable.
For a \(\rho\)-invariant subspace \(W\) of \(V_{\mathfrak{h}}=M(1)\otimes \mathbb{C}[\mathfrak{h}]\), abuse the notation \(W^{\pm}\) for the \(\pm 1\)-eigenspaces of \(W\) under \(\rho\).
\begin{lemma}
    \label{lemma10Modules}
    The following 10 spaces are irreducible  \(V_{L_2}^{S_4}\) modules,
    \[
        ((V^+_{\mathbb{Z}\beta})^0)^+, \ ((V^+_{\mathbb{Z}\beta})^0)^-,
    \]
    \[
        (V^-_{\mathbb{Z}\beta})^+, \ (V^-_{\mathbb{Z}\beta})^-,
    \]
    \[
        (V_{\mathbb{Z}\beta+\frac{1}{8} \beta })^+, \ (V_{\mathbb{Z}\beta+\frac{1}{8} \beta })^-,
    \]
    \[
        (V_{\mathbb{Z}\beta+\frac{3}{8} \beta })^+, \ (V_{\mathbb{Z}\beta+\frac{3}{8} \beta })^-,
    \]
    \[
        (V_{\mathbb{Z}\beta+\frac{1}{4} \beta }^0)^+, \ ((V_{\mathbb{Z}\beta+\frac{1}{4} \beta }^0)^-.
    \]
\end{lemma}
\begin{proof}
    The order of \(\rho\) is 2. So, directly employ Theorem \ref{MT1}.
\end{proof}

\begin{proposition}
    By Reference \cite{Dong2015476}, as irreducible \(V_{L_2}^{A_4}\) modules,
    \[
        V_{\mathbb{Z}\gamma+\frac{1}{18}\gamma}\cong W^{1,T_{1},0},\ V_{\mathbb{Z}\gamma-\frac{1}{18}\gamma}\cong W^{1,T_{2},0},
    \]
    \[
        V_{\mathbb{Z}\gamma+\frac{2}{18}\gamma}\cong W^{2,T_{2},0},\ V_{\mathbb{Z}\gamma-\frac{2}{18}\gamma}\cong W^{2,T_{1},0},
    \]
    \[
        V_{\mathbb{Z}\gamma+\frac{4}{18}\gamma}\cong W^{2,T_{1},1},\ V_{\mathbb{Z}\gamma-\frac{4}{18}\gamma}\cong W^{2,T_{2},1},
    \]
    \[
        V_{\mathbb{Z}\gamma+\frac{5}{18}\gamma}\cong W^{1,T_{2},1},\ V_{\mathbb{Z}\gamma-\frac{5}{18}\gamma}\cong W^{1,T_{1},1},
    \]
    \[
        V_{\mathbb{Z}\gamma+\frac{7}{18}\gamma}\cong W^{1,T_{1},2},\ V_{\mathbb{Z}\gamma-\frac{7}{18}\gamma}\cong W^{1,T_{2},2},
    \]
    \[
        V_{\mathbb{Z}\gamma+\frac{8}{18}\gamma}\cong W^{2,T_{2},2},\ V_{\mathbb{Z}\gamma-\frac{8}{18}\gamma}\cong W^{2,T_{1},2},
    \]
\end{proposition}

\noindent Let \(H\)=\(\langle \sigma \rangle\). Then, \(H\subseteq S_4\) satisfies the assumption in Theorem \ref{MT3}.
Remark \ref{Mod0} shows that \(V_{\mathbb{Z}\gamma \pm \frac{r}{18}}\), for \( r\in \mathbb{Z}, \ 1\leq r \leq 8\), and \(r \neq 0 (\mathrm{mod} 3)\) are twelve irreducible \(V_{L_2}^{\langle \sigma \rangle}\) modules of type two.
Theorem \ref{MT1} indicates that the twelve irreducible \(V_{L_2}^{\langle \sigma \rangle}\) modules of type two are twelve irreducible \(V_{L_2}^{S_4}\) modules of type two, probably isomorphic under \(V_{L_2}^{S_4}\), of type two. \\
\\
\noindent As \(V_{L_2}^{D_3}\cong V_{\mathbb{Z}\gamma}^{+}\) irreducible modules, lemma \ref{adtype} shows that
\[
    V_{\mathbb{Z}\gamma+\frac{1}{18}\gamma}\cong V_{\mathbb{Z}\gamma-\frac{1}{18}\gamma},
    V_{\mathbb{Z}\gamma+\frac{2}{18}\gamma}\cong V_{\mathbb{Z}\gamma-\frac{2}{18}\gamma},
    V_{\mathbb{Z}\gamma+\frac{4}{18}\gamma}\cong V_{\mathbb{Z}\gamma-\frac{4}{18}\gamma},
\]
\[
    V_{\mathbb{Z}\gamma+\frac{5}{18}\gamma}\cong V_{\mathbb{Z}\gamma-\frac{5}{18}\gamma},
    V_{\mathbb{Z}\gamma+\frac{7}{18}\gamma}\cong V_{\mathbb{Z}\gamma-\frac{7}{18}\gamma},
    V_{\mathbb{Z}\gamma+\frac{8}{18}\gamma}\cong V_{\mathbb{Z}\gamma-\frac{8}{18}\gamma},
\]
\\
\noindent Notice that \(D_3 \subset S_4\), that is, \(V_{L_2}^{S^4}\) is a vertex operator subalgebra of \(V_{L_2}^{D_3}\). Thus, as \(V_{L_2}^{S_4}\) irreducible modules,
\[
    V_{\mathbb{Z}\gamma+\frac{1}{18}\gamma}\cong V_{\mathbb{Z}\gamma-\frac{1}{18}\gamma},
    V_{\mathbb{Z}\gamma+\frac{2}{18}\gamma}\cong V_{\mathbb{Z}\gamma-\frac{2}{18}\gamma},
    V_{\mathbb{Z}\gamma+\frac{4}{18}\gamma}\cong V_{\mathbb{Z}\gamma-\frac{4}{18}\gamma},
\]
\[
    V_{\mathbb{Z}\gamma+\frac{5}{18}\gamma}\cong V_{\mathbb{Z}\gamma-\frac{5}{18}\gamma},
    V_{\mathbb{Z}\gamma+\frac{7}{18}\gamma}\cong V_{\mathbb{Z}\gamma-\frac{7}{18}\gamma},
    V_{\mathbb{Z}\gamma+\frac{8}{18}\gamma}\cong V_{\mathbb{Z}\gamma-\frac{8}{18}\gamma},
\]
\begin{lemma}
    \label{lemma6Pairs}
    Definition of \(g\) stable module shows that these twelve irreducible \(V_{L_2}^{\langle \sigma \rangle}\) modules are not \(\rho\) stable as and
    \[
        V_{\mathbb{Z}\gamma+\frac{1}{18}\gamma}^{\rho}\cong V_{\mathbb{Z}\gamma-\frac{1}{18}\gamma},
        V_{\mathbb{Z}\gamma+\frac{2}{18}\gamma}^{\rho}\cong V_{\mathbb{Z}\gamma-\frac{2}{18}\gamma},
        V_{\mathbb{Z}\gamma+\frac{4}{18}\gamma}^{\rho}\cong V_{\mathbb{Z}\gamma-\frac{4}{18}\gamma},
    \]
    \[
        V_{\mathbb{Z}\gamma+\frac{5}{18}\gamma}^{\rho}\cong V_{\mathbb{Z}\gamma-\frac{5}{18}\gamma},
        V_{\mathbb{Z}\gamma+\frac{7}{18}\gamma}^{\rho}\cong V_{\mathbb{Z}\gamma-\frac{7}{18}\gamma},
        V_{\mathbb{Z}\gamma+\frac{8}{18}\gamma}^{\rho}\cong V_{\mathbb{Z}\gamma-\frac{8}{18}\gamma},
    \]
\end{lemma}

\begin{remark}
    \label{remarkExactly1}
    Actions of \(\omega\) on the first level of a module show that those six irreducible \(V_{L_2}^{S_4}\) modules are not isomorphic.
    Theorem \ref{MT2} indicates that those six irreducible \(V_{L_2}^{\langle \sigma \rangle}\) modules exhaust irreducible \(V_{L_2}^{S_4}\) modules occurring in irreducible \(\sigma^i\), for \(i=1,2\), twisted \(V_{L_2}\) modules.
    So, those six irreducible \(V_{L_2}^{S_4}\) modules are exactly the six irreducible \(V_{L_2}^{S_4}\) modules occurring in irreducible \(\sigma^i\), for \(i=1,2\), twisted \(V_{L_2}\) modules.
\end{remark}

\begin{proposition}
    \label{propNinaG}
    By Reference \cite{Dong2015476}, let \(g\) be an automorphism of \(V_{L_2}\) of order \(T\neq 1\). Then there exists some vector \(u\in (V_{L_2})_1\), such that \(g=e^{2\pi i u(0)}\).
\end{proposition}
\noindent The action of \(\{e^{2\pi i h(0)}|h\in (V_{L_2})_{1}\}\) on \(V_{L_2}\) is isomorphic to \(SO(3)\), and on \(V_{\mathbb{Z}\frac{\alpha}{2}}\) is isomorphic to \(SU(2)\).
The action of the group generated by \(\{ \sigma, \tau_{1}, \tau_{2}, \tau_{3}, \rho\}\) on \(V_{L_2}\) is isomorphic to \(S_4\). Proposition \ref{propNinaG} shows that \(\langle \sigma, \tau_{1}, \tau_{2}, \tau_{3}, \rho \rangle\) is a subgroup of \(\{e^{2\pi i h(0)}|h\in (V_{L_2})_{1}\}\).
So, \(\langle \sigma, \tau_{1}, \tau_{2}, \tau_{3}, \rho \rangle\) acts on \(V_{\mathbb{Z}\frac{\alpha}{2}}=M(1)\otimes \mathbb{C}[\frac{1}{2}\mathbb{Z}\alpha]\).
Theorem \ref{theoremSchur} is the action of the group \(\langle \sigma, \tau_{1}, \tau_{2}, \tau_{3}, \rho \rangle\) on \(V_{\mathbb{Z}\frac{\alpha}{2}}\) is a Schur cover of \(S_4\), which is isomorphic to \(GL(2,3)\), the general linear group of degree 2 over a field of three elements. Thus, by the quantum Galois theory \cite{dong1997}
\begin{equation}
    \label{qdEquation}
    V_{\mathbb{Z}\frac{\alpha}{2}}\cong \bigoplus_{\chi}V_{\chi}\otimes W_{\chi}
\end{equation}
\noindent , where \(\chi\) runs over all irreducible characters of \(GL(2,3)\).
The irreducible representations of the group \(GL(2,3)\) are well known, two 1-dimensional, three 2-dimensional, two 3-dimensional, and one 4-dimensional irreducible representations.
Denote them by \(W_{1}^{i},W_2^{j},W_3^{k}, W_{4}\), where \(i=0, 1,\ j=0, 1, 2, \ k=0, 1\).
The subindex is the dimensional of the module and the upper indices distinguish the irreducible modules of the same dimension. \\
\\
Reference \cite{Dong2015476} shows that
\begin{equation}
    \label{equation1}
    V_{L_2}=(V_{\mathbb{Z}\beta}^+)^{0}\otimes U_{1}^{0} \oplus (V_{\mathbb{Z}\beta}^+)^{1}\otimes U_{1}^{1} \oplus (V_{\mathbb{Z}\beta}^+)^{2}\otimes U_{1}^{2}\oplus
    V_{\mathbb{Z}\beta}^{-}\otimes U_{3},
\end{equation}
\noindent and
\begin{equation}
    \label{equation2}
    V_{L_2+\frac{1}{2} \alpha}=V_{\mathbb{Z}\beta+\frac{1}{4}\beta}^{0}\otimes U_{2}^{0} \oplus V_{\mathbb{Z}\beta+\frac{1}{4}\beta}^{1}\otimes U_{2}^{1} \oplus V_{\mathbb{Z}\beta+\frac{1}{4}\beta}^{2}\otimes U_{2}^{2}
\end{equation}

\begin{lemma}
    \label{lemma2Pairs}
    The following isomorphisms hold,
    \[
        ((V_{\mathbb{Z}\beta}^+)^{1})^{\rho}\cong (V_{\mathbb{Z}\beta}^+)^{2},
    \]
\noindent and
    \[
        (V_{\mathbb{Z}\beta+\frac{1}{4}\beta}^{1})^{\rho} \cong V_{\mathbb{Z}\beta+\frac{1}{4}\beta}^{2}.
    \]
\end{lemma}
\begin{proof}
    Lemma \ref{lemma10Modules} shows that, \((V_{\mathbb{Z}\beta}^+)^{0}\) in equation \ref{equation1} becomes \(((V_{\mathbb{Z}\beta}^+)^{0})^{+}\), and \(((V_{\mathbb{Z}\beta}^+)^{0})^{-}\) in equation \ref{qdEquation} for \(V_{L_2}^{S_4}\).
    Likewise, \(V_{\mathbb{Z}\beta}^{-}\) in equation \ref{equation1} becomes \((V_{\mathbb{Z}\beta}^{-})^{+}\), and \((V_{\mathbb{Z}\beta}^{-})^{-}\) in equation \ref{qdEquation} for \(V_{L_2}^{S_4}\).
    Also, \(V_{\mathbb{Z}\beta+\frac{1}{4}\beta}^{0}\) in Equation \ref{equation2} becomes \((V_{\mathbb{Z}\beta+\frac{1}{4}\beta}^{0})^{+}\), and \((V_{\mathbb{Z}\beta+\frac{1}{4}\beta}^{0})^{-}\) in Equation \ref{qdEquation} for \(V_{L_2}^{S_4}\). At this point, there is a total of six nonisomorphic irreducible \(V_{L_2}^{S_4}\) modules in Equation \ref{qdEquation}
     Notice that total number of irreducible characters of \(GL(2,3)\) is eight.
     Quantum Galois theory \ref{qdEquation} shows that the total number of nonisomorphic irreducible \(V_{L_2}^{S_4}\) modules in Equation \ref{qdEquation} is eight. Theorem \ref{Qd1} shows that
     \[
        \mathrm{qdim}_{V_{L^2}^{A_4}}(V_{\mathbb{Z}\beta}^+)^{1}=\mathrm{qdim}_{V_{L^2}^{A_4}}(V_{\mathbb{Z}\beta}^+)^{2}=1
     \]
     and
     \[
        \mathrm{qdim}_{V_{L^2}^{A_4}}V_{\mathbb{Z}\beta+\frac{1}{4}\beta}^{1}=\mathrm{qdim}_{V_{L^2}^{A_4}}V_{\mathbb{Z}\beta+\frac{1}{4}\beta}^{2}=2
     \]
     This forces \(((V_{\mathbb{Z}\beta}^+)^{1})^{\rho}\cong (V_{\mathbb{Z}\beta}^+)^{2}\), and \((V_{\mathbb{Z}\beta+\frac{1}{4}\beta}^{1})^{\rho} \cong V_{\mathbb{Z}\beta+\frac{1}{4}\beta}^{2}\).
\end{proof}

\begin{remark}
    Lemma \ref{lemma2Pairs}, Equation \ref{qdEquation}, Equation \ref{equation1}, Equation \ref{equation2} show that
    \begin{equation}
        \label{equation3}
        \begin{aligned}
            V_{L_2}=&((V_{\mathbb{Z}\beta}^+)^{0})^{+}\otimes W_{1}^{0} \oplus ((V_{\mathbb{Z}\beta}^+)^{0})^{-}\otimes W_{1}^{1} \oplus (V_{\mathbb{Z}\beta}^+)^{1}\otimes W_{2}^{0}\\
            &\oplus
            (V_{\mathbb{Z}\beta}^{-})^{+}\otimes W_{3}^{0}\oplus
            (V_{\mathbb{Z}\beta}^{-})^{-}\otimes W_{3}^{1},
        \end{aligned}
    \end{equation}
and
    \begin{equation}
        \label{equation4}
        V_{L_2+\frac{1}{2} \alpha}=(V_{\mathbb{Z}\beta+\frac{1}{4}\beta}^{0})^{+}\otimes W_{2}^{1} \oplus (V_{\mathbb{Z}\beta+\frac{1}{4}\beta}^{0})^{-}\otimes W_{2}^{2} \oplus V_{\mathbb{Z}\beta+\frac{1}{4}\beta}^{1}\otimes W_{4}.
    \end{equation}
\end{remark}

\begin{lemma}
    \label{lemma8modules}
    The following 8 spaces are irreducible  \(V_{L_2}^{S_4}\) modules
     \[
        V_{\mathbb{Z}\gamma+\frac{r}{18}\gamma},\ \mathrm{for} \ (s\in \mathbb{Z}, 1\leq r \leq 8,\ \mathrm{and}\ r \neq 0 \ (\mathrm{mod}\ 3),
    \]
    \[
        (V^+_{\mathbb{Z}\beta})^1, V_{\mathbb{Z}\beta+\frac{1}{4} \beta }^1.
    \]
\end{lemma}
\begin{proof}
    The desired result follows from Lemma \ref{constructionLemma2}, Lemma \ref{lemma2Pairs}, and Remark \ref{remarkExactly1}.
\end{proof}
\begin{remark}
    \label{remarkConclusion1}
    Lemmas \ref{lemma10Modules} and \ref{lemma8modules} provide eighteen irreducible \(V_{L_2}^{S_4}\) modules, which exhaust irreducible \(V_{L_2}^{S_4}\) modules occurring in irreducible \(V_{L_2}^{A_4}\) modules.
    Lemmas \ref{lemma10Modules}, \ref{lemma6Pairs}, and \ref{lemma2Pairs} show that there are exactly five irreducible \(\rho\) stable \(V_{L_2}^{A_4}\) modules.
    Thus, by Theorem \ref{minvariance}, there are exactly five irreducible \(\rho\) twisted \(V_{L_2}^{A_4}\) modules.
    So, Lemma \ref{constructionLemma1}, Theorems \ref{MT1} and \ref{MT2} indicate that there are exactly ten extra irreducible \(V_{L_2}^{S_4}\) modules except for the eighteen  irreducible \(V_{L_2}^{S_4}\) modules occurring in irreducible \(V_{L_2}^{A_4}\) modules.
\end{remark}

\noindent Consider the vertex operator algebra isomorphism
\[
    V_{L_2}^{S_4}\cong (V_{\mathbb{Z}\beta}^{+})^{\langle \sigma, \rho \rangle} \cong (V_{\mathbb{Z}\beta}^{+})^{D_3}.
\]
Theorem \ref{MT2} and Remark \ref{remarkConjugate} show that irreducible \(V_{L_2}^{S_4}\) modules come from
\begin{itemize}
    \item irreducible \(V_{\mathbb{Z}\beta}^{+}\) modules,
    \item irreducible \(\sigma\) twisted \(V_{\mathbb{Z}\beta}^{+}\) modules,
    \item irreducible \(\sigma^2\) twisted \(V_{\mathbb{Z}\beta}^{+}\) modules,
    \item irreducible \(\rho\) twisted \(V_{\mathbb{Z}\beta}^{+}\) modules.
\end{itemize}

\begin{remark}
    Let \(M\) be an irreducible \(V_{\mathbb{Z}\beta}^{+}\) module.
    Then, Theorem \ref{MT1} and Lemma \ref{constructionLemma2} shows that \(M\) is a \(V_{L_2}^{A_4}\) module, and a direct sum of irreducible \(V_{L_2}^{A_4}\) modules.
    So, irreducible \(V_{L_2}^{S_4}\) modules from irreducible \(V_{\mathbb{Z}\beta}^{+}\) modules are irreducible \(V_{L_2}^{S_4}\) modules from irreducible \(V_{L_2}^{A_4}\) modules.
\end{remark}
\begin{remark}
    \label{remarkConclusion2}
    Let \(M\) be an irreducible \(\sigma\) twisted \(V_{\mathbb{Z}\beta}^{+}\) module. Then, Lemmas \ref{constructionLemma1}, \ref{constructionLemma2} shows that \(M\) is a \(V_{L_2}^{A_4}\) module, and hence a direct sum of irreducible \(V_{L_2}^{A_4}\) modules.
    So, irreducible \(V_{L_2}^{S_4}\) modules from irreducible \(\sigma\) twisted \(V_{\mathbb{Z}\beta}^{+}\) modules are irreducible \(V_{L_2}^{S_4}\) modules from irreducible \(V_{L_2}^{A_4}\) modules.
    Likewise, irreducible \(V_{L_2}^{S_4}\) modules from irreducible \(\sigma^2\) twisted \(V_{\mathbb{Z}\beta}^{+}\) modules are irreducible \(V_{L_2}^{S_4}\) modules from irreducible \(V_{L_2}^{A_4}\) modules.
    Therefore, those ten extra irreducible \(V_{L_2}^{S_4}\) modules in Remark \ref{remarkConclusion1} are from irreducible \(\rho\) twisted \(V_{\mathbb{Z}\beta}^{+}\) modules.
\end{remark}
\noindent Notice that Lemma \ref{adtype} shows that
\[
    (V_{\mathbb{Z}\beta}^{+})^{\langle \rho \rangle}\cong V_{\mathbb{Z}\zeta}^{+}\cong V_{L_{2}}^{D_4}.
\]
\noindent Reference \cite{dong1999representations} shows that the lattice vertex operator algebra \(V_{n\mathbb{Z}\alpha}^{+}\) is generated by \(\omega, J, \ \mathrm{and}\ E_{n\alpha}\), where
\[
    J=(x^{1}(-1))^4\mathbf{1}-2x^{1}(-3)h(-1)\mathbf{1}+\frac{3}{2}(x^{1}(-2))^2\mathbf{1}
\]
\[
    E_{n\alpha}=e^{n\alpha}+e^{-n\alpha}.
\]
\\
\noindent Definition of \(\rho\) shows that
\[
    \rho(E_{2\alpha})=-E_{2\alpha},
\]
\noindent and
\[
    \rho(E_{4\alpha})=E_{4\alpha}.
\]
\noindent Notice that \(V_{\mathbb{Z}\beta}^{+}\) is generated by \(\omega, J, \ \mathrm{and}\  E_{2\alpha}\), and \(V_{\mathbb{Z}\zeta}^{+}\) is generated by \(\omega, J, \ \mathrm{and}\ E_{4\alpha}\).
Thus, \((V_{\mathbb{Z}\beta}^{+})^{\langle \rho \rangle}\) and \(V_{\mathbb{Z}\zeta}^{+}\) share the same generators.
That is, they are not only isomorphic, but also the same vertex operator algebra.

\begin{lemma}
    \label{lemma2TwistedModules}
    As \(V_{\mathbb{Z}\zeta}^{+}\) modules,
    \[
        V_{\mathbb{Z}\beta}^{T_{1},+} \cong (V_{\mathbb{Z}\beta}^{T_{2},+})^{\rho} \cong V_{\mathbb{Z}\zeta}^{T_{1},+},
    \]
    and
    \[
        V_{\mathbb{Z}\beta}^{T_{2},-} \cong (V_{\mathbb{Z}\beta}^{T_{2},-})^{\rho} \cong V_{\mathbb{Z}\zeta}^{T_{1},-}.
    \]
\end{lemma}
\begin{proof}
     Since \(\zeta=2\beta\), \(\mathbb{Z}\zeta\) acts as 0 on \(\mathbb{Z}\beta/(2\mathbb{Z}\beta)\).
     Thus, definition of  \(V_{\mathbb{Z}\beta}^{T_{i}}\), where \(i=1,2\), shows that \(\mathbb{Z}\zeta\) acts as 1 on both \(V_{\mathbb{Z}\beta}^{T_{1}}\) and \(V_{\mathbb{Z}\beta}^{T_{2}}\).
     So, as \(V_{\mathbb{Z}\zeta}^{+}=(V_{\mathbb{Z}\beta}^{+})^{\langle \rho \rangle}\) modules,
     \[
        (V_{\mathbb{Z}\beta}^{T_{2},+})^\rho\cong  V_{\mathbb{Z}\zeta}^{T_{1},+}, \ V_{\mathbb{Z}\beta}^{T_{1},+} \cong  V_{\mathbb{Z}\zeta}^{T_{1},+},
     \]
and
     \[
        (V_{\mathbb{Z}\beta}^{T_{2},-})^\rho\cong  V_{\mathbb{Z}\zeta}^{T_{1},-}, \ V_{\mathbb{Z}\beta}^{T_{1},-} \cong  V_{\mathbb{Z}\zeta}^{T_{1},-}.
     \]
     Transitivity of congruence yields the desired results.
\end{proof}

\begin{lemma}
    \label{lemma8Twistedmodules}
    The following eight spaces are irreducible \(V_{\mathbb{Z}\zeta}^{+}=(V_{\mathbb{Z}\beta}^{+})^{\langle \rho \rangle}\) modules occurring in irreducible \(\rho\) twisted \(V_{\mathbb{Z}\beta}^{+}\) modules.
    \[
            V_{\mathbb{Z}\zeta+\frac{s}{32}\zeta},\ \mathrm{for}\ s\in \mathbb{Z}, 1\leq s \leq 15,\ \mathrm{and}\ s \neq 0 \ \mathrm{mod}\ 2),
    \]
\end{lemma}
\begin{proof}
    Reference \cite{dong1999representations} shows that these eight spaces are irreducible \(V_{\mathbb{Z}\zeta}^{+}=(V_{\mathbb{Z}\beta}^{+})^{\langle \rho \rangle}\) modules.
    Let \(M\) be an irreducible \(V_{\mathbb{Z}\zeta}^{+}=(V_{\mathbb{Z}\beta}^{+})^{\langle \rho \rangle}\) module from an irreducible \(V_{\mathbb{Z}\beta}^{+}\) module.
    Then, the action of \(\omega\) on the first level of \(M\) is \(\lambda+n\), where \(\lambda\) is the action of \(\omega\) on the first level of the irreducible \(V_{\mathbb{Z}\beta}^{+}\) module, and \(n\) is a nonnegative integer.
    Check the action of \(\omega\) on the first level of each module listed in the lemma. None of them satisfies this condition.
    Theorem \ref{MT2} shows that an irreducible \(V_{\mathbb{Z}\zeta}^{+}=(V_{\mathbb{Z}\beta}^{+})^{\langle \rho \rangle}\) modules is from irreducible \(\rho\) twisted \(V_{\mathbb{Z}\beta}^{+}\) modules, or from irreducible \(V_{\mathbb{Z}\beta}^{+}\) modules.
    Since these eight modules are not from irreducible \(V_{\mathbb{Z}\beta}^{+}\) modules, they are from irreducible \(\rho\) twisted \(V_{\mathbb{Z}\beta}^{+}\) modules.
\end{proof}
\begin{remark}
    \label{remark2TwistedModules}
    Use Theorem \ref{Qd1} to check the quantum dimensions of irreducible \(V_{\mathbb{Z}\beta}^{+}\) modules, and of \(V_{\mathbb{Z}\zeta}^{+}\).
    The four irreducible \(V_{\mathbb{Z}\zeta}^{+}\) modules, \(\{V_{\mathbb{Z}\zeta}^{T_{1},\pm}, V_{\mathbb{Z}\zeta}^{T_{2},\pm} \}\), are either from \(\{V_{\mathbb{Z}\beta}^{T_{1},\pm}, V_{\mathbb{Z}\beta}^{T_{2},\pm}\}\), or from \(\rho\) irreducible twisted \(V_{\mathbb{Z}\beta}^{+}\) modules.
    Lemma \ref{lemma2TwistedModules} shows that \(V_{\mathbb{Z}\zeta}^{T_{1},\pm}\) are from \(\{V_{\mathbb{Z}\beta}^{T_{1},\pm}, V_{\mathbb{Z}\beta}^{T_{2},\pm}\}\). Hence, \(V_{\mathbb{Z}\zeta}^{T_{2},\pm}\) are from irreducible \(\rho\) twisted \(V_{\mathbb{Z}\beta}^{+}\) modules.
\end{remark}
\begin{remark}
    \label{remark10TwistedModules}
    Lemma \ref{lemma8Twistedmodules} and Remark \ref{remark2TwistedModules} provide ten irreducible \(V_{\mathbb{Z}\zeta}^{+}\) occurring in irreducible \(\rho\) twisted \(V_{\mathbb{Z}\beta}^{+}\) modules.
    Let \(M\) be an irreducible \(\rho\) twisted \(V_{\mathbb{Z}\beta}^{+}\) module.
    Notice that \(C_{D_3}(\rho)=\langle \rho \rangle \), and \( \langle \rho \rangle \subset G_M\). The fact \(G_M\subset C_{D_3}(\rho)\) shows that \( \langle \rho \rangle = G_M \).
    Theorem \ref{MT1} indicates that the ten irreducible \(V_{\mathbb{Z}\zeta}^{+}\) modules are ten irreducible \(V_{L_2}^{S_4}\) modules, probably isomorphic under \(V_{L_2}^{S_4}\), occurring in irreducible \(\rho\) twisted \(V_{\mathbb{Z}\beta}^{+}\) modules.
\end{remark}
\begin{remark}
    \label{remarkConclusion3}
    Actions of \(\omega\) on the first level of a module show that those ten irreducible \(V_{L_2}^{S_4}\) modules in Remark \ref{remark10TwistedModules} are not isomorphic.
    Theorem \ref{MT2} indicates that those ten irreducible \(V_{\mathbb{Z}\zeta}^{+}\) modules exhaust irreducible \(V_{L_2}^{S_4}\) modules occurring in irreducible \(\rho\) twisted \(V_{\mathbb{Z}\beta}^{+}\) modules.
    So, those ten irreducible \(V_{L_2}^{S_4}\) modules are the exactly ten irreducible \(V_{L_2}^{S_4}\) modules occurring in irreducible \(\rho\) twisted \(V_{\mathbb{Z}\beta}^{+}\) modules.
\end{remark}
\begin{theorem}
   Irreducible modules of \(V_{L_2}^{S_4}\) are
    \[
        ((V^+_{\mathbb{Z}\beta})^0)^+, \ ((V^+_{\mathbb{Z}\beta})^0)^-,
    \]
    \[
        (V^-_{\mathbb{Z}\beta})^+, \ (V^-_{\mathbb{Z}\beta})^-,
    \]
    \[
        (V_{\mathbb{Z}\beta+\frac{1}{8} \beta })^+, \ (V_{\mathbb{Z}\beta+\frac{1}{8} \beta })^-,
    \]
    \[
        (V_{\mathbb{Z}\beta+\frac{3}{8} \beta })^+, \ (V_{\mathbb{Z}\beta+\frac{3}{8} \beta })^-,
    \]
    \[
        (V_{\mathbb{Z}\beta+\frac{1}{4} \beta }^0)^+, \ (V_{\mathbb{Z}\beta+\frac{1}{4} \beta }^0)^-,
    \]
    \[
        (V^+_{\mathbb{Z}\beta})^1, V_{\mathbb{Z}\beta+\frac{1}{4} \beta }^1,
    \]
    \[
        V_{\mathbb{Z}\gamma+\frac{r}{18}\gamma},\ \mathrm{for} \ r\in \mathbb{Z}, 1\leq r \leq 8,\ \mathrm{and}\ r \neq 0 \ (\mathrm{mod}\ 3),
    \]
    \[
        V_{\mathbb{Z}\zeta+\frac{s}{32}\zeta},\ \mathrm{for}\ s\in \mathbb{Z}, 1\leq s \leq 15,\ \mathrm{and}\ s \neq 0\ (\mathrm{mod}\ 2),
    \]
    \[
        V_{\mathbb{Z}\zeta}^{T_2,+},\ V_{\mathbb{Z}\zeta}^{T_2,-}.
    \]
\end{theorem}
\begin{proof}
    The desired result follows from Remark \ref{remarkConclusion1}, Remark \ref{remarkConclusion2}, and Remark \ref{remarkConclusion3}.
\end{proof}

\begin{theorem} The quantum dimensions for all irreducible $V_{L}^{S_4}$-modules over $V_{L}^{S_4}$
are given by the following tables.
\begin{center}
\begin{tabular}{|c|c|c|c|c|c|}
\hline
 & \(((V^+_{\mathbb{Z}\beta})^0)^+\) & \(((V^+_{\mathbb{Z}\beta})^0)^{-}\) & \((V^+_{\mathbb{Z}\beta})^1\)& \((V^-_{\mathbb{Z}\beta})^{+}\)& \((V^-_{\mathbb{Z}\beta})^{-}\)\tabularnewline
\hline
$\mathrm{qdim}$ & 1 & 1 & 2 & 3 & 3\tabularnewline
\hline
$$ & \(M^{0}\) & \(M^{1}\) & \(M^{2}\) & \(M^{3}\)& \(M^{4}\)\tabularnewline
\hline
\end{tabular}
\par\end{center}

\begin{center}
\begin{tabular}{|c|c|c|c|}
\hline
 & \((V_{\mathbb{Z}\beta+\frac{1}{4} \beta }^0)^+\) & \((V_{\mathbb{Z}\beta+\frac{1}{4} \beta }^0)^{-}\) & \(V_{\mathbb{Z}\beta+\frac{1}{4} \beta }^1\)\tabularnewline
\hline
$\mathrm{qdim}$ & 2 & 2 & 4\tabularnewline
\hline
$$ & \(M^{6}\) & \(M^{7}\)& \(M^{8}\)\tabularnewline
\hline
\end{tabular}
\par\end{center}

\begin{center}
\begin{tabular}{|c|c|c|c|c|}
\hline
 & \((V_{\mathbb{Z}\beta+\frac{1}{8} \beta })^+\) & \((V_{\mathbb{Z}\beta+\frac{1}{8} \beta })^-\) & \((V_{\mathbb{Z}\beta+\frac{3}{8} \beta })^+\)& \((V_{\mathbb{Z}\beta+\frac{3}{8} \beta })^{-}\)\tabularnewline
\hline
$\mathrm{qdim}$ & 6 & 6 & 6 & 6\tabularnewline
\hline
$$ & \(M^{9}\) & \(M^{10}\)& \(M^{11}\) & \(M^{12}\) \tabularnewline
\hline
\end{tabular}
\par\end{center}

\begin{center}
\begin{tabular}{|c|c|c|c|c|}
\hline
 & \(V_{\mathbb{Z}\gamma+\frac{r}{18}\gamma}\) & \(V_{\mathbb{Z}\zeta+\frac{s}{32}\zeta}\) & \(V_{\mathbb{Z}\zeta}^{T_2,+}\)& \(V_{\mathbb{Z}\zeta}^{T_2,-}\)\tabularnewline
\hline
$\mathrm{qdim}$ & 8 & 6 & 12 & 12\tabularnewline
\hline
$$ & \(M^{13},\ldots, M^{20}\) & \(M^{21}, \ldots, M^{26}\)& \(M^{27}\) & \(M^{28}\) \tabularnewline
\hline
\end{tabular}
\par\end{center}
In the last table, \(r,s\in \mathbb{Z}\), \(1\leq r\leq 8\), \(1\leq s\leq 15\), \(r\neq 0\) (mod 3), and \(s\neq 0\) (mod 2).
\end{theorem}

\begin{proof}
    The definition of \(V_L^{T_i,\pm}\), for \(i=1,2\), indicates
    \[
        \mathrm{qdim}_{V_{L}^{+}}V_L^{T_i,\pm}=\mathrm{qdim}_{V_{L}^{+}}V_{\mathbb{Z}\alpha}.
    \]
    Hence, the quantum dimensions listed are obtained by Theorem \ref{Qd1} and Theorem \ref{Qd2}.
\end{proof}

\section{Irreducible modules of \(V_{L_2}^{A_5}\)}
\noindent Let \(\beta=2\alpha\), \(\gamma=3\alpha\), and \(\mu=5\alpha\).
\begin{remark}
    \label{remarkA5Group}
    \noindent A collection of properties of the alternating group \(A_5\) are given. \
    \begin{itemize}
        \item (a) \(A_5\) is simple.
        \item (b)Subgroups of \(A_5\) of a fixed order has a unique conjugacy class.
        \item (c)The maximal proper subgroups of \(A_5\) are isomorphic to \(S_3\), \(D_5\), and \(A_4\).
        \item (d)The Sylow 2 group is isomorphic to the Klein four group \(K\). The Sylow 3 group is isomorphic to the cyclic group \(\mathbb{Z}_3\). The Sylow 5 group is isomorphic to the cyclic group \(\mathbb{Z}_5\).
        \item (e)The projective special linear group of degree two for \(A\) is \(PSL(2,5)\). The corresponding Schur cover is \(SL(2,5)\).
    \end{itemize}
\end{remark}
\begin{remark}
    Denote the corresponding groups in \(A\) by \(V_4\), \(\mathbb{Z}/3\mathbb{Z}\), \(\mathbb{Z}/5\mathbb{Z}\), twisted \(S_3\), \(D_5\), and \(A_4\). Use \(C_G(H)\) and \(N_G(H)\) to denote centralizer and normalizer respectively. Then,
    \[
         C_{A_5}(V_4)=V_4,\ C_{A_5}(\mathbb{Z}/3\mathbb{Z})=\mathbb{Z}/3\mathbb{Z},\ C_{A_5} (\mathbb{Z}/5\mathbb{Z})=\mathbb{Z}/5\mathbb{Z},
    \]
and
    \[
        N_{A_5}(V_4)=A_4, \ N_{A_5}(\mathbb{Z}/3\mathbb{Z})=S_3,\ N_{A_5}(\mathbb{Z}/5\mathbb{Z})=D_5,
    \]
\end{remark}

\begin{remark}
    There are two irreducible \(V_{L_2}\) modules, \(V_{L_2}\) and \(V_{L_2+\frac{\alpha}{2}}\).
    The action of \(\omega\) on the first level of \(V_{L_2}\) is 0, and on \(V_{L_2+\frac{\alpha}{2}}\) is \(\frac{1}{4}\).
    Thus, Definition \ref{defStableModule} shows that both \(V_{L_2}\) and \(V_{L_2+\frac{\alpha}{2}}\) are \(g\) stable, for each \(g\in A_5\).
    Let \(J\) be a subgroup of \(A_5\). Remark \ref{remarkType} shows that there are two types of irreducible \(V_{L_2}^{J}\) modules modules.
    \begin{itemize}
        \item An irreducible \(V_{L_2}^{J}\) module \(M\) is of \emph{type one} if \(M\) occurs in the decomposition of \(V_{\frac{\mathbb{Z}\alpha}{2}}\), as \(V_{L_2}^{J}\) modules.
        \item An irreducible \(V_{L_2}^{J}\) module \(M\) is of \emph{type two} if \(M\) does not occur in the decomposition of \(V_{\frac{\mathbb{Z}\alpha}{2}}\), as \(V_{L_2}^{J}\) modules. That is, \(M\) occurs in a \(h\) twisted \(V_{L_2}\) module for some \(h\in J\) and \(h\neq 1\).
    \end{itemize}
\end{remark}

\begin{remark}
    \label{remark3Groups}
    Use Theorem \ref{theoremV_L^+}, Lemma \ref{lemmaK}, Lemma \ref{adtype}, and Remark \ref{Mod0}.
    \begin{itemize}
        \item Irreducible \(V_{L_2}^{V_4}\) modules of type two are \(V_{\mathbb{Z}\beta+\frac{1}{8}\beta}\), \(V_{\mathbb{Z}\beta+\frac{3}{8}\beta}\), \(V_{\mathbb{Z}\beta}^{T_1, \pm}\), and \(V_{\mathbb{Z}\beta}^{T_2, \pm}\).
        \item Irreducible \(V_{L_2}^{\mathbb{Z}/3\mathbb{Z}}\) modules of type two are \(V_{\mathbb{Z}\gamma \pm \frac{r}{18}\gamma}\), for \( r\in \mathbb{Z}, \ 1\leq r \leq 17\), and \(r \neq 0 \ (\mathrm{mod}\ 3)\).
        \item Irreducible \(V_{L_2}^{\mathbb{Z}/5\mathbb{Z}}\) modules of type two are \(V_{\mathbb{Z}\mu \pm \frac{r}{18}\gamma}\), for \(t\in \mathbb{Z}, \ 1\leq t \leq 49\), and \(t \neq 0 \ (\mathrm{mod}\ 5)\).
    \end{itemize}
\end{remark}
\begin{lemma}
    \label{lemmaA5Key1}
    Let \(H\in \{V_4, \mathbb{Z}/3\mathbb{Z}, \mathbb{Z}/5\mathbb{Z}\} \). Then, an irreducible \(V_{L_2}^{H}\) module of type two is an irreducible \(V_{L_2}^{A_5}\) module of type two.
\end{lemma}
\begin{proof}
    Notice that \(H\subseteq A_5\) satisfies the assumption in Theorem \ref{MT3}. The desired result follows from Theorem \ref{MT3}.
\end{proof}

\begin{lemma}
    \label{lemmaA5Key2}
    Let \(H\in \{V_4, \mathbb{Z}/3\mathbb{Z}, \mathbb{Z}/5\mathbb{Z}\} \). Then, nonisomorphic irreducible \(V_{L_2}^{N_{A_5}(H)}\) modules of type two are nonisomorphic irreducible \(V_{L_2}^{A_5}\) modules of type two.
\end{lemma}
\begin{proof}
    Consider \(H=V_4\). Then, \(N_{A_5}(H)=A_4\). Reference \cite{Dong201376} shows that those six irreducible \(V_{L_2}^{V_4}\) modules of type two in Remark \ref{remarkA5Group} become two nonisomorphic irreducible \(V_{L_2}^{A_4}\) modules, \(V_{\mathbb{Z}\beta+\frac{1}{8}\beta}\), \(V_{\mathbb{Z}\beta+\frac{3}{8}\beta}\). Lemma \ref{lemmaA5Key1} shows that \(V_{\mathbb{Z}\beta+\frac{1}{8}\beta}\) and  \(V_{\mathbb{Z}\beta+\frac{3}{8}\beta}\) are irreducible \(V_{L_2}^{A_5}\) modules. The action of \(\omega\) on the first levels of these two modules are distinct. Thus, they are nonisomorphic irreducible \(V_{L_2}^{A_5}\) modules. \\
    \\
    Consider \(H=\mathbb{Z}/3\mathbb{Z}\). Then, \(N_{A_5}(H)=S_3\). Theorem \ref{theoremV_L^+} shows that those twelve irreducible \(V_{L_2}^{\mathbb{Z}/3\mathbb{Z}}\) modules of type two in Remark \ref{remarkA5Group} become six nonisomorphic irreducible \(V_{L_2}^{S_3}\) modules, \(V_{\mathbb{Z}\gamma \pm \frac{r}{18}\gamma}\), for \( r\in \mathbb{Z}, \ 1\leq r \leq 8\), and \(r \neq 0 (\mathrm{mod} 3)\). Lemma \ref{lemmaA5Key1} shows that those six modules are irreducible \(V_{L_2}^{A_5}\) modules. The action of \(\omega\) on the first levels of these six modules are distinct. Thus, they are nonisomorphic irreducible \(V_{L_2}^{A_5}\) modules. \\
    \\
    Consider \(H=\mathbb{Z}/5\mathbb{Z}\). Then, \(N_{A_5}(H)=D_5\). Theorem \ref{theoremV_L^+} shows that those forty irreducible \(V_{L_2}^{\mathbb{Z}/5\mathbb{Z}}\) modules of type two in Remark \ref{remarkA5Group} become twenty nonisomorphic irreducible \(V_{L_2}^{D_5}\) modules, \(V_{\mathbb{Z}\mu \pm \frac{t}{50}\mu}\), for \( t\in \mathbb{Z}, \ 1\leq t \leq 24\), and \(r \neq 0 \ (\mathrm{mod}\ 5)\). Lemma \ref{lemmaA5Key1} shows that those twenty modules are irreducible \(V_{L_2}^{A_5}\) modules. The action of \(\omega\) on the first levels of these twenty modules are distinct. Thus, they are nonisomorphic irreducible \(V_{L_2}^{A_5}\) modules.
\end{proof}

\begin{lemma}
    \label{lemmaA5Key3}
    Irreducible \(V_{L_2}^{A_5}\) modules of type two are those twenty eight modules constructed in the proof of Lemma \ref{lemmaA5Key2}.
\end{lemma}
\begin{proof}
    The desired result follows from Lemma \ref{lemmaA5Key1}, Lemma \ref{lemmaA5Key2}, Theorem \ref{theoremConjugate}, and Remark \ref{remarkA5Group}.
\end{proof}

\noindent Proposition \ref{propNinaG} shows that \(A_5\) can be considered as a subgroup of \(\{e^{2\pi i h(0)}|h\in (V_{L_2})_{1}\}\).
Thus, \(A_5\) acts on \(V_{\mathbb{Z}\frac{\alpha}{2}}=M(1)\otimes \mathbb{C}[\frac{1}{2}\mathbb{Z}\alpha]\).
Remark \ref{remarkSchurA5} shows that the action of the group \(A_5\) on \(V_{\mathbb{Z}\frac{\alpha}{2}}\) is a Schur cover of \(A_4\), which is isomorphic to \(SL(2,5)\), the general linear group of degree 2 over a field of three elements. Thus, by the quantum Galois theory \cite{dong1997}
\begin{equation}
    \label{qdEquation1}
    V_{\mathbb{Z}\frac{\alpha}{2}}\cong \bigoplus_{\chi}V_{\chi}\otimes W_{\chi}
\end{equation}
\noindent , where \(\chi\) runs over all irreducible characters of \(SL(2,5)\).
The irreducible representations of the group \(SL(2,5)\) are well known, one 1-dimensional, two 2-dimensional, two 3-dimensional, two 4-dimensional, one 5-dimensional, and one 6-dimensional irreducible representations.
Denote them by \(X_1,X_2^{i},X_3^{j}, X_{k}, X_5, X_6\), where \(i=0, 1,\ j=0, 1, \ k=0, 1\).
The subindex is the dimension of the module and the upper indices distinguish the irreducible modules of the same dimension.
\begin{remark}
    \label{remarkKeyA5}
    Irreducible \(V_{L_2}^{A_5}\) modules of type one occurs and exhausts the modules in the complete decomposition in Equation \ref{qdEquation1}.
    By quantum Galois Theory, express these modules as \(T_1,T_2^{i},T_3^{j}, T_4^{k}, T_5, T_6\), where \(i=0, 1,\ j=0, 1, \ k=0, 1\).
    The subindex is the quantum dimension of the module and the upper indices distinguish the irreducible modules of the same dimension.
\end{remark}

\begin{theorem}
    Assume that \(V_{L_2}^{A_5}\) is rational and $C_2$ cofiniteness. There are thirty seven irreducible \(V_{L_2}^{A_5}\) modules.
    \begin{itemize}
        \item \(T_1,T_2^{i},T_3^{j}, T_4^{k}, T_5, T_6\), where \(i=0, 1,\ j=0, 1, \ k=0, 1\),
        \item \(V_{\mathbb{Z}\beta+\frac{1}{8}\beta}\), \(V_{\mathbb{Z}\beta+\frac{3}{8}\beta}\),
        \item \(V_{L_2}^{S_3}\) modules, \(V_{\mathbb{Z}\gamma \pm \frac{r}{18}}\gamma\), for \( r\in \mathbb{Z}, \ 1\leq r \leq 8\), and \(r \neq 0 \ (\mathrm{mod}\ 3)\),
        \item \(V_{L_2}^{D_5}\) modules, \(V_{\mathbb{Z}\mu \pm \frac{t}{50}\mu}\), for \( t\in \mathbb{Z}, \ 1\leq t \leq 24\), and \(r \neq 0 \ (\mathrm{mod}\ 5)\),
    \end{itemize}
\end{theorem}
\begin{proof}
    The desired results follows from Lemma \ref{lemmaA5Key3} and remark \ref{remarkKeyA5}.
\end{proof}

\begin{theorem} The quantum dimensions for all irreducible $V_{L_2}^{A_5}$-modules over $V_{L_2}^{A_5}$
are given by the following tables.

\begin{center}
\begin{tabular}{|c|c|c|c|c|c|c|}
\hline
 & \(T_1\) & \(T_2^i\) & \(T_3^i\)& \(T_4^k\)&\(T_5\)&\(T_6\)\tabularnewline
\hline
$\mathrm{qdim}$ & 1 & 2 & 3 & 4 & 5 & 6 \tabularnewline
\hline
\end{tabular}
\par\end{center}

\begin{center}
\begin{tabular}{|c|c|c|c|c|}
\hline
 & \(V_{\mathbb{Z}\beta+\frac{1}{8}\beta}\) & \(V_{\mathbb{Z}\beta+\frac{3}{8}\beta}\) & \(V_{\mathbb{Z}\gamma \pm \frac{r}{18}\gamma}\)& \(V_{\mathbb{Z}\mu \pm \frac{t}{50}\mu}\)\tabularnewline
\hline
$\mathrm{qdim}$ & 30 & 30 & 20 & 12\tabularnewline
\hline
\end{tabular}
\par\end{center}

\noindent In the first table, \(i=0, 1,\ j=0, 1, \ k=0, 1\). In the second table, \(r,s\in \mathbb{Z}\), \(1\leq r\leq 8\), \(1\leq s\leq 15\), \(r\neq 0\) (mod 3), and \(s\neq 0\) (mod 2).
\end{theorem}

\begin{proof}
    The quantum dimensions listed are obtained by Remark \ref{remarkKeyA5}, Theorem \ref{Qd1}, and Theorem \ref{Qd2}.
\end{proof}

\begin{remark}
    These classification of \(V_{L_2}^{A_5}\) is based on the \(C_2\) cofiniteness and on the rationality of \(V_{L_2}^{A_5}\).
    However, Theorem \ref{MT3} requires the rationality of \(V^{H}\) rather than the rationality of \(V^{G}\).
    For any proper subset \(H\) of \(A_5\), notice that \(V^{H}\) is rational.
    Therefore, even if \(V_{L_2}^{A_5}\) was not rational, the irreducible modules listed in this paper would still be correct irreducible modules, but might not exhaust all of them.
\end{remark}
\bibliographystyle{plain}
\bibliography{LiuyiZhangRef}

\begin{thebibliography}{10}

\bibitem{abe2000fusion}
Toshiyuki Abe.
\newblock Fusion rules for the free bosonic orbifold vertex operator algebra.
\newblock {\em Journal of Algebra}, 229(1):333--374, 2000.

\bibitem{abe2005rationality}
Toshiyuki Abe.
\newblock Rationality of the vertex operator algebra {$V_L^+$} for a positive
  definite even lattice {$L$}.
\newblock {\em Mathematische Zeitschrift}, 249(2):455--484, 2005.

\bibitem{abe2004rationality}
Toshiyuki Abe, Geoffrey Buhl, and Chongying Dong.
\newblock Rationality, regularity, and {$C_2$}-cofiniteness.
\newblock {\em Transactions of the American Mathematical Society},
  356(8):3391--3402, 2004.

\bibitem{borcherds1986vertex}
Richard~E Borcherds.
\newblock Vertex algebras, kac-moody algebras, and the monster.
\newblock {\em Proceedings of the National Academy of Sciences},
  83(10):3068--3071, 1986.

\bibitem{dong2011characterization2011arXiv1110.1882D}
C.~{Dong} and C.~{Jiang}.
\newblock {A characterization of vertex operator algebras
  {$V_{\mathbb{Z}\alpha}^{+}$}: I}.
\newblock {\em ArXiv e-prints}, October 2011.

\bibitem{2015arXiv150703306DOrbifold}
C.~{Dong}, L.~{Ren}, and F.~{Xu}.
\newblock {On Orbifold Theory}.
\newblock {\em ArXiv e-prints}, July 2015.

\bibitem{dong1998rank}
Chongying Dong and Robert~L Griess.
\newblock Rank one lattice type vertex operator algebras and their automorphism
  groups.
\newblock {\em Journal of Algebra}, 208(1):262--275, 1998.

\bibitem{dong2014rationality}
Chongying Dong and Jianzhi Han.
\newblock On rationality of vertex operator superalgebras.
\newblock {\em International Mathematics Research Notices},
  2014(16):4379--4399, 2014.

\bibitem{dong2010characterization}
Chongying Dong and Cuipo Jiang.
\newblock A characterization of vertex operator algebra
  {$L(\frac{1}{2},0)\otimes L(\frac{1}{2},0)$}.
\newblock {\em Communications in Mathematical Physics}, 296(1):69--88, 2010.

\bibitem{dong2013characterization}
Chongying Dong and Cuipo Jiang.
\newblock A characterization of the rational vertex operator algebra: {II}.
\newblock {\em Advances in Mathematics}, 247:41--70, 2013.

\bibitem{Dong201376}
Chongying Dong and Cuipo Jiang.
\newblock Representations of the vertex operator algebra {$V_{L_2}^{A_4}$}.
\newblock {\em Journal of Algebra}, 377:76 -- 96, 2013.

\bibitem{dong2014characterization}
Chongying Dong and Cuipo Jiang.
\newblock A characterization of the vertex operator algebra {$V_{L_2}^{A_4}$}.
\newblock In {\em Conformal Field Theory, Automorphic Forms and Related
  Topics}, pages 55--74. Springer, 2014.

\bibitem{Dong2015476}
Chongying Dong, Cuipo~(Cuibo) Jiang, Qifen Jiang, Xiangyu Jiao, and Nina Yu.
\newblock Fusion rules for the vertex operator algebra {$V_{L_2}^{A_4}$}.
\newblock {\em Journal of Algebra}, 423:476 -- 505, 2015.

\bibitem{dong2013quantum}
Chongying Dong, Xiangyu Jiao, and Feng Xu.
\newblock Quantum dimensions and quantum galois theory.
\newblock {\em Transactions of the American Mathematical Society},
  365(12):6441--6469, 2013.

\bibitem{dong1997regularity}
Chongying Dong, Haisheng Li, and Geoffrey Mason.
\newblock Regularity of rational vertex operator algebras.
\newblock {\em Advances in Mathematics}, 132(1):148--166, 1997.

\bibitem{dong1998twisted}
Chongying Dong, Haisheng Li, and Geoffrey Mason.
\newblock Twisted representations of vertex operator algebras.
\newblock {\em Mathematische Annalen}, 310(3):571--600, 1998.

\bibitem{dong2000modular}
Chongying Dong, Haisheng Li, and Geoffrey Mason.
\newblock Modular-invariance of trace functions in orbifold theory and
  generalized moonshine.
\newblock {\em Communications in Mathematical Physics}, 214(1):1--56, 2000.

\bibitem{dong2005elliptic}
Chongying Dong, Kefeng Liu, and Xiaonan Ma.
\newblock Elliptic genus and vertex operator algebras.
\newblock {\em Pure and Applied Mathematics Quarterly}, 1(4), 2005.

\bibitem{dong1997}
Chongying Dong and Geoffrey Mason.
\newblock On quantum galois theory.
\newblock {\em Duke Math. J.}, 86(2):305--321, 1997.

\bibitem{dong2004rational}
Chongying Dong and Geoffrey Mason.
\newblock Rational vertex operator algebras and the effective central charge.
\newblock {\em International Mathematics Research Notices},
  2004(56):2989--3008, 2004.

\bibitem{Dong1999384}
Chongying Dong and Kiyokazu Nagatomo.
\newblock {Classification of Irreducible Modules for the Vertex Operator
  Algebra} {$M(1)^+$}.
\newblock {\em Journal of Algebra}, 216(1):384 -- 404, 1999.

\bibitem{dong1999representations}
Chongying Dong and Kiyokazu Nagatomo.
\newblock Representations of vertex operator algebra {$V_{L}^{+}$} for rank one
  lattice {$L$}.
\newblock {\em Communications in mathematical physics}, 202(1):169--195, 1999.

\bibitem{dong2001classification}
Chongying Dong and Kiyokazu Nagatomo.
\newblock Classification of irreducible modules for the vertex operator algebra
  {$M(1)^+$}: {II}. higher rank.
\newblock {\em Journal of Algebra}, 240(1):289--325, 2001.

\bibitem{dong1994twisted}
CY~Dong.
\newblock Twisted modules for vertex algebras associated with even lattices.
\newblock {\em Journal of Algebra}, 165(1):91--112, 1994.

\bibitem{frenkel1993axiomatic}
Igor Frenkel, Yi-Zhi Huang, and James Lepowsky.
\newblock {\em On axiomatic approaches to vertex operator algebras and
  modules}, volume 494.
\newblock American Mathematical Soc., 1993.

\bibitem{frenkel1989vertex}
Igor Frenkel, James Lepowsky, and Arne Meurman.
\newblock {\em Vertex operator algebras and the Monster}, volume 134.
\newblock Academic press, 1989.

\bibitem{ginsparg1988curiosities}
P~Ginsparg.
\newblock Curiosities at {$c=1$}.
\newblock {\em Nuclear Physics B}, 295(2):153--170, 1988.

\bibitem{hanaki1999quantum}
Akihide Hanaki, Masahiko Miyamoto, Daisuke Tambara, et~al.
\newblock Quantum galois theory for finite groups.
\newblock {\em Duke mathematical journal}, 97(3):541--544, 1999.

\bibitem{MR1205350}
P.~N. Hoffman and J.~F. Humphreys.
\newblock {\em Projective representations of the symmetric groups}.
\newblock Oxford Mathematical Monographs. The Clarendon Press, Oxford
  University Press, New York, 1992.
\newblock $Q$-functions and shifted tableaux, Oxford Science Publications.

\bibitem{huang2008vertex}
Yi-Zhi Huang.
\newblock Vertex operator algebras and the verlinde conjecture.
\newblock {\em Communications in Contemporary Mathematics}, 10(01):103--154,
  2008.

\bibitem{kiritsis1989proof}
Elias~B Kiritsis.
\newblock Proof of the completeness of the classification of rational conformal
  theories with {$c=1$}.
\newblock {\em Physics Letters B}, 217(4):427--430, 1989.

\bibitem{miyamoto2011Flat&Semi-Rigidity}
M.~{Miyamoto}.
\newblock {Flatness and Semi-Rigidity of Vertex Operator Algebras}.
\newblock {\em ArXiv e-prints}, April 2011.

\bibitem{miyamoto1998representation}
Masahiko Miyamoto.
\newblock Representation theory of code vertex operator algebra.
\newblock {\em Journal of Algebra}, 201(1):115--150, 1998.

\bibitem{miyamoto2009flatness}
Masahiko Miyamoto.
\newblock Flatness of tensor products and semi-rigidity for {$C_2$}-cofinite
  vertex operator algebras i.
\newblock {\em arXiv preprint arXiv:0906.1407}, 2009.

\bibitem{miyamoto2015c_2OfCyclicCMP}
Masahiko Miyamoto.
\newblock {$C_2$}-cofiniteness of cyclic-orbifold models.
\newblock {\em Communications in Mathematical Physics}, 335(3):1279--1286,
  2015.

\bibitem{schur2001representation}
J~Schur.
\newblock On the representation of the symmetric and alternating groups by
  fractional linear substitutions.
\newblock {\em International Journal of Theoretical Physics}, 40(1):413--458,
  2001.

\bibitem{verlinde1988fusion}
Erik Verlinde.
\newblock Fusion rules and modular transformations in {2D} conformal field
  theory.
\newblock {\em Nuclear Physics B}, 300:360--376, 1988.

\bibitem{zbMATH05622792}
Robert~A. {Wilson}.
\newblock {\em {The finite simple groups.}}
\newblock London: Springer, 2009.

\bibitem{xu2000algebraic}
Feng Xu.
\newblock Algebraic orbifold conformal field theories.
\newblock {\em Proceedings of the National Academy of Sciences},
  97(26):14069--14073, 2000.

\bibitem{zhang2009w}
Wei Zhang and Chongying Dong.
\newblock {$W$}-algebra {$W(2, 2)$} and the vertex operator algebra
  {$L(\frac{1}{2},0)\otimes L(\frac{1}{2},0)$}.
\newblock {\em Communications in Mathematical Physics}, 285(3):991--1004, 2009.

\bibitem{zhu1996modular}
Yongchang Zhu.
\newblock Modular invariance of characters of vertex operator algebras.
\newblock {\em Journal of the American Mathematical Society}, 9(1):237--302,
  1996.

\end{thebibliography}
\end{document}